\documentclass{article}
\usepackage[a4paper,hmargin=1.2in,vmargin=1.2in]{geometry} 
\usepackage{xspace,url}
\usepackage{latexsym}
\usepackage{amssymb}
\usepackage{amsmath}
\usepackage{amsfonts}
\usepackage{graphicx}
\newcommand\p{\mathbf{p}}
\newcommand\q{\mathbf{q}}
\newcommand\te{\mathbf{t}}
\newcommand\es{\mathbf{s}}

\newcommand\er{\mathbf{r}}
\newcommand\R{\mathbf{R}}

\newcommand\conc{\mathcal{\circ}}
\newcommand{\lift}{}
\newcommand{\creas}{c}

\newcommand\zmax{z_\text{max}}
\newcommand\fw{face-weight\xspace}
\newcommand\fws{face-weights\xspace}
\usepackage{xspace}
\usepackage{stmaryrd}
\usepackage{booktabs}
\graphicspath{{figs/}}

\newtheorem{lem}{Lemma}

\newtheorem{defn}{Definition}
\newtheorem{thm}{Theorem}
\newtheorem{cor}{Corollary}

\author{Erik D. Demaine\thanks{Computer Science and Artificial Intelligence
Laboratory, MIT, Cambridge, MA. {\tt edemaine@mit.edu}}
\and
Andr\'e Schulz\thanks{LG Theoretische Informatik, FernUniversit\"at in Hagen, {\tt
      andre.schulz@fernuni-hagen.de}. Supported by
the German Research Foundation (DFG) under grant SCHU 2458/2-1. }}

\title{Embedding Stacked Polytopes on a Polynomial-Size Grid}

\begin{document}

 \maketitle

\begin{abstract}
A stacking operation adds a $d$-simplex on top of a facet of a simplicial $d$-polytope while 
maintaining the convexity of the polytope. A stacked $d$-polytope is a polytope that
is obtained from a $d$-simplex and a series of stacking operations. We show that
for a fixed $d$ every stacked $d$-polytope with $n$ vertices can be realized with nonnegative 
integer coordinates. The coordinates are bounded by $O(n^{2\log_2(2d)})$, except
for one axis, where the coordinates are bounded by $O(n^{3\log_2(2d)})$. The
described realization can be computed with an easy algorithm.

The realization of the polytopes is obtained with a lifting technique which produces
an embedding on a large grid. We establish 
a rounding scheme that places the vertices on a sparser grid, while maintaining the
convexity of the embedding.
\end{abstract}
\section{Introduction}

Steinitz's Theorem~\cite{S22,Z95} states that the graphs of 3-polytopes%
\footnote{In our terminology a \emph{polytope} is always understood as a \emph{convex} polytope.}
are exactly the planar 3-connected graphs.
In particular, every planar 3-connected graph can be realized as a 3-polytope.
The original proof is constructive, transforming the graph by a sequence of
local operations down to a tetrahedron.  Unfortunately, the resulting
polytope construction requires exponentially many bits of accuracy for each
vertex coordinate.  Stated another way, this construction can place the
$n$ vertices on an integer grid, but that grid may have dimensions
doubly exponential in~$n$~\cite{OS94}. The situation in higher dimensions is even worse.
Already in dimension $4$, there are polytopes that cannot be realized
with rational coordinates, and a 4-polytope that can be realized
on the grid might require coordinates that are doubly exponential in
the number of its vertices \cite{Z95}.
Moreover, it is $\sf NP$-hard to decide whether a lattice is a face lattice
of a 4-polytope~\cite{RG96,GZ95}.

How large an integer grid do we need to embed a given planar 3-connected graph
with $n$ vertices as a polytope?  This question goes back at least eighteen years
as Problem~4.16 in G\"unter M. Ziegler's book \cite{Z95};
he wrote that ``it is quite possible that there is a quadratic upper bound''
on the length of the maximum dimension.
The best bound so far is exponential in $n$, namely $O(2^{7.21n})$
\cite{BS10,RRS11}; see below for the long history.
The central question is whether a polynomial grid suffices, that is,
whether Steinitz's Theorem can be made efficient.
For comparison, a planar graph can be embedded in the plane with strictly
convex faces using a polynomial-size grid \cite{BR06}.
In this paper, we give the first nontrivial subexponential upper bound
for a large class of 3-polytopes. Moreover, our construction generalizes to higher
dimensions and we show that a nontrivial class of $d$-polytopes can 
be realized with integers coordinates, which are bounded by a polynomial
in $n$.

A \emph{$d$-dimensional stacked polytope} is a polytope that is constructed by a sequence
of ``stacking operations'' applied to a $d$-simplex. 
 A \emph{stacking operation}
glues a $d$-simplex $\Delta$ atop a simplicial facet $f$ of  polytope, by identifying
$f$ with a face of $\Delta$, while maintaining the convexity of the
polytope.  Thus a stacking operation removes one facet $f$ and adds $d$ new
facets having a new common vertex.
We call this new vertex \emph{stacked on $f$}.
%
%

\paragraph{Our results.}
We present an algorithm that realizes a stacked polytope on a grid whose
dimensions are polynomial in~$n$.  In our presentation $\log$ denotes the binary logarithm.
Our main result is the following:

\newcommand{\mainthm}{
\begin{thm}\label{thm:main}
Every $d$-dimensional stacked polytope can be realized on an integer grid,
such that all coordinates 
have size at most $10 d^2  R^2$, except for one axis, where the coordinates have size at most $6 R^3$, for
$R= d n^{\log(2d)}$.
\end{thm}
}
\mainthm
As a corollary of Theorem~\ref{thm:main} we obtain that every stacked 3-polytope can be embedded
on the grid of dimensions $270 n^{5.17}\times 270 n^{5.17} \times 18 n^{7.76}$.

\paragraph{Related work.}

Several algorithms have been developed to realize a given graph as a
3-polytope.
Most of these algorithms are based on the following two-stage approach.
The first stage computes a plane (\emph{flat}) embedding.  To extend the plane drawing
to a 3-polytope, the plane drawing must fulfill a criterion which can
be phrased as an ``equilibrium stress condition''.  Roughly speaking,
replacing every edge of the graph with a spring, the resulting system of
springs must be in a stable state for the plane embedding.
Plane drawings that fulfill this criterion for the interior vertices
can be computed as barycentric embeddings, i.e.,
by Tutte's method~\cite{T60,T63}.
The main difficulty is to guarantee the equilibrium condition for the boundary
vertices as well, because in general this goal is achievable only for certain
locations of the outer face.
The second stage computes a 3-polytope by assigning every vertex a height
expressed in terms of the spring constants of the system of springs.

The two-stage approach finds application in a
series of algorithms~\cite{CGT96,EG95,HK92,OS94,RRS11,RG96,Sch11}. The first result that improves the induced grid embedding of Steinitz's construction is due to Onn and Sturmfels~\cite{OS94}; they achieved a grid size of $O(n^{160n^3})$. 
Richter-Gebert's algorithm \cite{RG96} uses a grid of size  $O(2^{18 n^2})$ for general 3-polytopes, and a grid of size $O(2^{5.43n})$ if the graph of the polytope contains
at least one triangle.
These bounds were improved by Rib\'o~Mor~\cite{R06} and later on  by
Rib\'o~Mor, Rote, and Schulz~\cite{RRS11}.
The last paper expresses an upper bound for the grid size
in terms of the number of spanning trees of the graph.
Using the recent bounds of Buchin and Schulz~\cite{BS10} on the number of
spanning trees, this approach gives an upper bound on the grid size of
$O(2^{7.21n})$ for general 3-polytopes and $O(2^{4.83n})$ for 3-polytopes
with at least one triangular face.
These bounds are the best known to date for the general case. Very recently, Pak and Wilson proved
that every simplicial\footnote{A polytope is \emph{simplicial} if all its faces are simplices.}
 3-polytope can be embedded on a grid of size $4n^3 \times 8n^5 \times (500n^8)^n$~\cite{PW13}.


Zickfeld showed in his PhD thesis~\cite{Z09} that it is possible to embed
very special cases of stacked 3-polytopes on a grid polynomial in $n$.
First, if each stacking operation takes place on one of the three
faces that were just created by the previous stacking operation
(what might be called a \emph{serpentine stacked polytope}), then there is an embedding on the
$n \times n \times 3n^4$ grid.
Second, if we perform the stacking in rounds, and in every round we
stack on every face simultaneously (what might be called the \emph{balanced} stacked polytope),
then there is an embedding on a
${4 \over 3} n \times {4 \over 3} n \times O(n^{2.47})$ grid.
Zickfeld's embedding algorithm for balanced stacked 3-polytopes constructs a barycentric embedding. Because of the
special structure of the underlying graph the plane embedding remarkably fits on a small grid.

Every stacked polytope can be extended to a balanced stacked polytope
at the expense of adding an exponential number of vertices.
By doing so, Zickfeld's grid embedding for the balanced case induces a
$O(2^{3.91n})$ grid embedding for general stacked 3-polytopes.

Little is known about the lower bound of the grid size for embeddings of 3-polytopes. 
An integral convex embedding of an $n$-gon in the plane needs
an area of $\Omega(n^{3})$~\cite{AZ95,A61,T91}.
Therefore, realizing a 3-polytope with an $(n-1)$-gonal face requires
at least one dimension of size $\Omega(n^{3/2})$. 
For simplicial polytopes (and hence for stacked polytopes), this lower-bound argument does not apply.
However, there are planar 3-trees that require  $\Omega(n^2)$ area for plane straight-line drawings~\cite{MNRA11}. We show in Sect.~\ref{sec:lower} how to get a lower bound of $\Omega(n^3)$ volume for stacked 3-polytopes based on this 2d example.


Alternative approaches for realizing general 3-polytopes come from the
original proof of Steinitz's theorem, as well as the Koebe--Andreev--Thurston
circle-packing theorem, which induces a particular polytope realization
called the \emph{canonical polytope} \cite{Z95}.
Das and Goodrich \cite{DG97} essentially perform many inverted edge contractions
on many independent vertices in one step,
resulting in a singly exponential bound on the grid size.
The proof of the Koebe--Andreev--Thurston circle-packing theorem relies on
nonlinear methods and makes the features of the 3-dimensional embedding obtained
from a circle packing intractable; see \cite{S73} for an overview.
Lov\'asz studied a method for realizing 3-polytopes using a vector
of the nullspace of a \emph{Colin de Verdi\`ere matrix} of rank~$3$~\cite{L00}.
It is easy to construct these matrices for stacked polytopes;
however, without additional requirements,
the computed grid embedding might again need an exponential-size grid.

\paragraph{Our contribution.}
At a high level, we follow the popular two-stage approach: we compute a \emph{flat embedding} in $\R^{d-1}$ 
and then lift it to a $d$-polytope.
Although this two-stage approach is well known for realizing 3-polytopes, it cannot be easily 
extended to higher dimensions. There exists a generalization of equilibrium stresses
for polyhedral complexes in higher dimensions by Rybnikov~\cite{R99}. However, it is not 
straightforward to operate with the formalism used by Rybnikov for our purposes.
To overcome the difficulties of handling the more complex behavior of 
the higher-dimensional polytopes and to keep our presentation self-contained, we develop our
own specialized  methods to study liftings of \emph{stacked} polytopes. Notice that our methods coincides 
with the approach of Rybnikov, however, we omit the proof, since it is not required for this presentation.

Instead of specifying
the stress and then computing the barycentric embedding we construct the ``stress''
and the barycentric embedding in parallel. To specify the stress we define the heights of
the lifting (actually, the vertical movement of the vertices as induced by the stacking operation). 
In our presentation the concept of stress is therefore reduced to a certificate for the convexity
of the lifting, but it is not defining the flat embedding. On the other hand we 
still use barycentric coordinates to determine the flat embedding. A crucial step in our algorithm
is the construction of a balanced set of barycentric coordinates, which corresponds
to face volumes in the flat embedding. Initially, all faces have the same volume, but to prevent large heights in the lifting, 
we increase the volumes of the small faces. To see which faces must be blown up, we make use of a 
decomposition technique from data structural analysis called \emph{heavy path decomposition} \cite{T83}. 
Based on this decomposition, we subdivide the stacked polytope 
into a hierarchy of (serpentine) stacked polytopes, which we use to define the barycentric coordinates.
At this stage the lifting of the flat embedding would result in a grid embedding with exponential
coordinates. But since we have balanced the volume assignments, we can 
allow a small perturbation of the embedding, while maintaining its convexity.
Analyzing the size of the feasible perturbations shows that we can round to 
points on a polynomially sized grid.

A preliminary version of this work was presented at the 22nd ACM-SIAM Symposium on Discrete Algorithms (SODA)
in  San Francisco~\cite{DS11}. In the preliminary version we were focused on the more prominent 
3-dimensional case and did not present any bounds for higher-dimensional
polytopes. For the sake of a unified presentation we changed the construction
of the lifting slightly. In the preliminary version we defined the constructed stress as a linear
combination of stresses defined on certain $K_4$s. In this paper we specify the ``movement''
for every stacked vertex (its vertical shift) directly. This can be considered
as the dual definition of the lifting. Another difference is the more careful analysis 
of the size of the $z$-coordinates in the final embedding. In contrast to the preliminary
version, where we presented bound of $224{,}000 n^{18}$ for the $z$-coordinates, we present
a different method for bounding the height of the  lifting, which 
yields an upper bound of $18 n^{7.76}$ instead. We remark also that the conference version contained a flaw in the area assignment, which is now fixed by a slightly modified construction (Lemma~\ref{lem:balance}). In a paper that followed the preliminary 
version Igamberdiev and Schulz introduced a duality transformation for 3-polytopes that allows to 
control  the grid size of the dual polytope~\cite{IS13}. By this our results for stacked polytopes can be transferred to
their dual polytopes (truncated 3-polytopes), which shows
that also this class can be realized on a polynomial-sized grid.

\section{Specifying the geometry of stacked polytopes}
In this section we develop the necessary tools for defining an embedding
of a stacked $d$-polytope with the two-stage approach. We present the framework
for our algorithm by defining its basic construction and by identifying 
certificates to guarantee its correctness. The actual algorithm is then presented
in Section~\ref{sec:algo}.

\subsection{Matrices, determinants, simplices}
We start our presentation with introducing some notation. Assume that
$S=(\es_1,\ldots , \es_k)$ is a sequence of $(k-1)$-dimensional vectors. Then we
denote by $(\es_1,\ldots , \es_k)$ the matrix whose row vectors (in order) are 
$\es_1,\ldots , \es_k$. We define as
$$[S]= \det \begin{pmatrix}  \es_1 & \es_2 & \hdots & \es_k  \\ 1 & 1 & \hdots
&1 \end{pmatrix}. $$
Notice that $[ S]$ equals $(k-1)!$ times the signed volume of the simplex spanned by
$\es_1,\ldots , \es_k$. When working with sequences we use the binary 
$\conc$ operator to concatenate two sequences. If there is no danger of confusion
we identify an element with the singleton set (or sequence) that contains this element. If $f$ is 
a function on $\mathcal{U}$ we extend $f$ in the natural way to act on 
 sequences on $\mathcal{U}$, i.e., we set $f((a_1,\ldots,a_z)):=(f(a_1),\ldots,f(a_z))$.

Throughout the paper we denote for a sequence $X$ of affinely independent points the
simplex spanned by $X$ by $\Delta_X$. If this simplex is the realization
of some face of an embedding of a polytope we denote this face by $f_X$.
For a $d$-polytope we call a $(d-1)$-face a \emph{facet}, and a $(d-2)$-face a \emph{ridge}.

By convention, we understand as $\R^k$ the space spanned by the first $k$ standard basis vectors. 
In this sense we consider the subspace $\R^{k-1}$ as space embedded inside $\R^k$.
Furthermore, when we speak about hyperplanes in the following, we consider only those hyperplanes in $\R^k$ that are not orthogonal to $\R^{k-1}$.


\subsection{Flat embedding}
An important intermediate step in our embedding algorithm is the construction of a
flat embedding, which is a simplicial complex constructed by repeated 
weighted barycentric subdivisions of a $(d-1)$-simplex $\Delta_B$. The 
combinatorial structure of the flat embedding represents the face lattice of the stacked polytope. 
Let $\Delta_S$ be a $(d-1)$-simplex and let $\Delta_1,\Delta_2,\ldots,\Delta_d$ be the
$(d-2)$-simplices that define its boundary. A \emph{barycentric subdivision}
adds a new point $\p$ in the interior of $\Delta_S$ and splits $\Delta_S$ into
$d$ $(d-1)$-simplices. For every simplex $\Delta_i$ we obtain a new simplex
$\Delta'_i$ spanned by   $\Delta_i\cup\{\p\}$. If  the
volume of all $\Delta'_i$s are given such that they sum up to the  volume of
$\Delta_S$, then there is one unique vertex $\p$, such that the subdivision
respects this volume assignment. We call the  volumes the $\Delta'_i$s the
\emph{barycentric coordinates} of $\p$ with respect to $\Delta_S$. 

We apply repeated barycentric subdivisions to create (a flat) embedding of
a stacked polytope. In particular, the subdivisions carry out the stacking
operations geometrically (projected in the $z=0$ hyperplane). We 
start the subdivision process with two copies of a $(d-1)$-simplex in the $z=0$ hyperplane
 glued along the boundary. One of these ``initial'' simplices will be the \emph{base face} $f_B$,
which remains unaltered during further modifications. The other face will be
repeatedly subdivided such that the combinatorial structure of the stacked polytope 
is obtained. The
constructed realization is still flat, which means that it lies in the 
$z=0$ hyperplane. 

\subsection{Liftings}

Our algorithm for realizing stacked polytopes is based on the \emph{lifting technique}.
The assignment of an additional coordinate to every vertex 
in the flat embedding is called a \emph{lifting}. For the assignment of a new
coordinate $z$ to the point $\p \in \R^{d-1}$ we use the shorthand notation $(\p,z)$.
The lifting is complemented by the
projection function $\pi\colon \R^d \to \R^{d-1}$ that simply removes the $d$-th coordinate.
As shortcut notation we write $\llbracket S
\rrbracket$ for the expression $[\pi(S)]$. In the following we refer to 
the $d$-th coordinate of a point in $\R^d$ as its $z$-coordinate.

A hyperplane $h\subset \R^d$ is
characterized by a function $z_h\colon \R^{d-1}\to \R$ that assigns to every
point $\p\in\R^{d-1}$ a new coordinate, such that $\lift(\p,z_h(\p))$ lies on
$h$. We denote by $h(S)$ the hyperplane \emph{spanned} by $S$ and use as shortcut
notation $z_S$ for $z_{h(S)}$. 
The following lemma gives us a convenient expression for the function $z_S$.
\begin{lem}\label{lem:altz}
Let $S$ be a set of $d$ affinely independent points from $\R^d$. 
For every point $\p\in \R^{d-1}$ we have
$$ z_S( \p)= \frac{[  S\conc  \lift(\p,0) ]\;}{\llbracket S \rrbracket}.$$

\end{lem}
\begin{proof}
We give a geometric proof of the lemma. Consider the $d$-simplex $\mathcal{P}_1$ spanned by
$\pi(S)$ and $(\p,z_S(\p))$, and the $d$-simplex $\mathcal{P}_2$ spanned by
$S$ and $(\p,0)$ as depicted in Fig.~\ref{fig:lemma1}. Simplex $\mathcal{P}_1$ is the
affine image of $\mathcal{P}_2$, under the mapping $(\q,z) \mapsto 
(\q,z_S(\q)-z)$. Since this mapping is a linear shear, it preserves the volumes of the simplices. 
The simplex $\mathcal{P}_1$ can be understood as a pyramid with base $\llbracket S \rrbracket$ and height $z_S(\p)$
and therefore its absolute volume $\mathit{vol}$ equals 
$$
	\mathit{vol}(\mathcal{P}_1)  =  \left\lvert \frac{1}{d(d-1)!} \llbracket S \rrbracket z_S(\p) \right\lvert.
$$
On the other hand we can express the absolute volume of $\mathcal{P}_2$ using the determinant formula and obtain 
$$
	\mathit{vol}(\mathcal{P}_2)  = \left\lvert \frac{1}{d!} [S \conc (\p,0)] \right\rvert. 
$$
Setting  $\mathit{vol}(\mathcal{P}_1)= \mathit{vol}(\mathcal{P}_2)$ proves $ z_S( \p)= \lvert [  S\conc  \lift(\p,0) ] / \llbracket S \rrbracket \rvert.$ 
It remains to check if the sign of the expression in the lemma is correct. 
By the geometric argument the correctness does not depend on the actual configuration given by $S$ and $\p$. To see this note that for a fixed $\p$ a sign change in $[S \conc (\p,0)]$ occurs if and only if it occurs in $\llbracket S \rrbracket$.
Moreover, if we keep $S$ fixed, a sign change of $[S \conc (\p,0)]$ can only happen
if $(\p,0)$ changes its location relative to $h(S)$, which on the other hand also changes
the sign for $z_S(\p)$.
Hence, it suffices to look at a concrete example. 
Let $S$ be given by the standard basis of $\R^d$, and let $\p$ be the origin
in $\R^{d-1}$. In this case, clearly, $[S \conc (\p,0)] = +1$ and  $\llbracket S \rrbracket = +1$. Since $z_S(\p)=1$ the sign is correct.
\end{proof}
\begin{figure}
\begin{center}
	\includegraphics[width=.8\columnwidth]{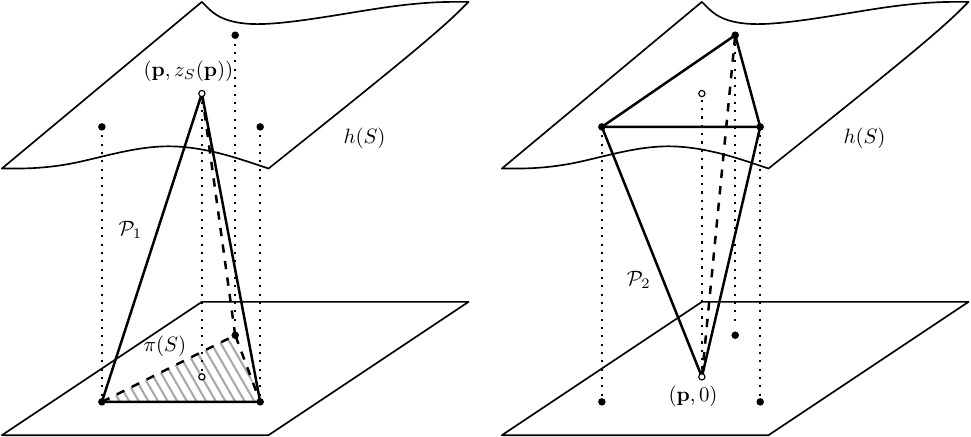}
\caption{The simplices $\mathcal{P}_1$ and $\mathcal{P}_2$ as defined in the proof of Lemma~\ref{lem:altz}. Both pyramids have the same (absolute) volume.}
\label{fig:lemma1}
\end{center}
\end{figure}

\subsection{Creases and stresses} 
Let $f,g$ be two hyperplanes in $\R^d$. Furthermore, let $S$, resp. $T$, be a
sequence of $d$ affinely independent points of $f$, resp. $g$, such that
deleting the last element in $S$ and $T$ gives the same subsequence  $X$. It
follows that $f$ and $g$ intersect in a flat of dimension $d-2$ that contains
the simplex $\Delta_X$. Furthermore, we denote by $\mathbf{s}_d$ the last point in
the sequence $S$, that is the point not in $T$, and similarly we denote by $\te_d$ 
the point in $T$ that is not in $S$. We also set $\er=\pi(\mathbf{s}_d)$. The situation is
depicted in Fig.~\ref{fig:crease}.
We define as \emph{creasing} on $X$
\begin{equation}\label{eq:defcrease}
\creas(f,g,X):= \frac{z_{T}(\er) - z_{S}(\er)}{\llbracket S \rrbracket}.  
\end{equation}
Roughly speaking, the creasing is a measure for the intersection angle of $f$
and $g$ with respect to the volume of  $\Delta_X$. 
\begin{figure}[htb]
\begin{center}

  \includegraphics[width=.50\columnwidth]{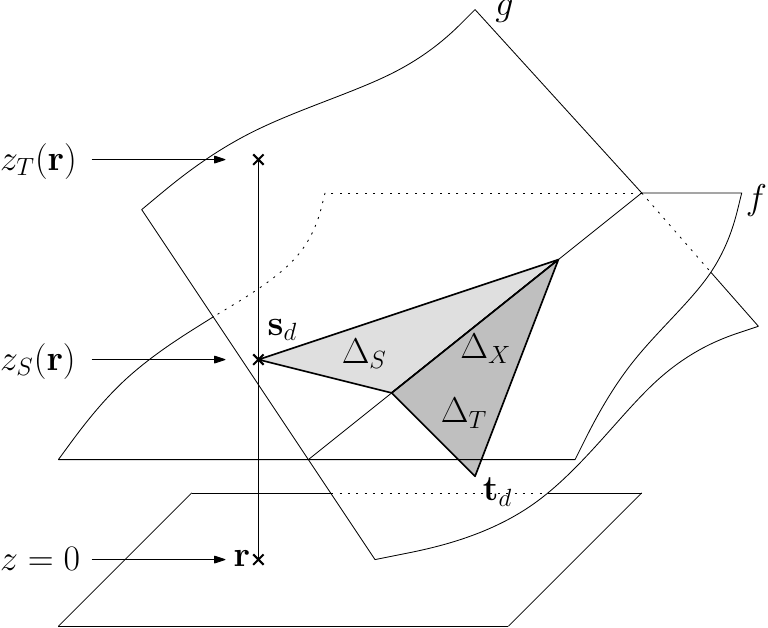} 
\caption{The creasing on $X$ is defined in terms of the objects depicted in the figure (for $d=3$). }
  \label{fig:crease}
\end{center}
\end{figure}

The following lemma shows that the function $c$ is well-defined.
\begin{lem}\label{lem:crease}
Let $c(f,g,X)$ be defined as above, i.e.,  $S$ spans $f$,  $T$
spans $g$, such that $S=X\conc \es_d$, and $T=X\conc \te_d$.

\begin{enumerate}
\item[(a)] It holds that $\creas(f,g,X)=-\creas(g,f,X)$.
\item[(b)] The value of $\creas(f,g,X)$ is independent of the choice of $\es_d$
and $\te_d$.
\end{enumerate}
\end{lem}
\begin{proof}
Let $T$ be the sequence $(\te_1,\ldots,\te_d)$ and let $z_i$ be the $z$-coordinate of $\te_i$. We use Lemma~\ref{lem:altz} to
rephrase $c(f,g,X)$ by substituting $z_T(\er)$ and obtain 
\begin{align}\label{eq:altstress0}
c(f,g,X) = \frac{z_{T}(\er) - z_{S}(\er)}{\llbracket S \rrbracket}& =\frac{[T\circ
\lift(\er,0)] - z_{S}(\er)\llbracket T \rrbracket}{\llbracket S \rrbracket \llbracket T \rrbracket}.
\end{align}
To further evaluate the last expression we notice that
$$ [T\conc \lift(\er,0)] - [T\conc \lift(\er,z_S(\er))] = z_S(\er)\llbracket T \rrbracket.$$
Since $[T \conc (\er, z_S(\er))]=  [T \conc \es_d]$, we obtain that $[T\circ \lift(\er,0)] - z_S(\er)\llbracket T\rrbracket = [T \conc \es_d]$.
Plugging the last equation into~\eqref{eq:altstress0} gives 
\begin{equation}\label{eq:altstress}
c(f,g,X)=\frac{z_{T}(\er) - z_{S}(\er)}{\llbracket S \rrbracket} =
\frac{[T\circ
\lift(\er,0)] - z_{S}(\er)\llbracket T \rrbracket}{\llbracket S \rrbracket \llbracket T \rrbracket} =
\frac{[T\conc \es_d]}{\llbracket T \rrbracket\llbracket S \rrbracket}. 
\end{equation}
Notice that $[T\conc \es_d]=-[S\conc \te_d]$, since both underlying
matrices differ only by one column swap (the sequence of the first
$d-1$ members 
of $S$ and $T$ coincide). As a consequence, exchanging $S$ and $T$ in the right
hand side of \eqref{eq:altstress} changes only the sign, which proves statement
(a).

To prove statement (b) we argue as follows. Assume we have replaced 
$\te_d$ in $T$ (that is the only point in $T$ not in $X$) by some other
point on $h(T)$. The replacement of $\te_d$
does not change $z_T(\er)$ in Equation~\eqref{eq:defcrease}, and all other parts
of the equation only depend on $S$. Therefore, the expression given in~\eqref{eq:defcrease}  
does not depend on $\te_d$. It remains to show, that also changing the
last point in $S$, does not change the creasing of $X$. However, this follows
easily since by statement (a) we can interchange the roles of $S$ and $T$ (this results in
a sign change), then we apply the above argument, and then we change $S$ and $T$
back (which cancels the sign change).
\end{proof}

Based on the concept of creasing we are now ready to define the crucial concept
of stress. Let us first introduce the notion of right and left facet. We assume 
that we have given a flat embedding at this point.
Let $X$ be the set of vertices of some ridge $f_X$ of $P$. We fix an arbitrary order for
$X$, and do so also for the vertex sets of the other ridges. The ridge $f_X$ separates
two facets of $P$, say $f_S$ and $f_T$. 
As usual, $S=X\cup{\es_d}$ is the vertex set of $f_S$. If $\llbracket X \conc \es_d \rrbracket >0$
then $f_S$ is called \emph{left of} $f_X$, if  $\llbracket X \conc \es_d \rrbracket <0$
then $f_S$ is called \emph{right of} $f_X$. We have one exception from this rule: for the 
special (base) face $f_B$ the notion of left and right is interchanged.
Note that the definitions left/right are independent
of the choice of the $z$-coordinates, and only depend on the flat embedding.
Note also, that every ridge has a right and a left incident facet.

\begin{defn}[stress]
Let $P$ be stacked $d$-polytope obtained by lifting a flat embedding and
let $X$ be a sequence of vertices which defines a ridge. We
define as \emph{stress on $X$}
\begin{equation}\label{eq:stressdef} \omega_X:= \creas(h_l,h_r,X),\end{equation}
where $h_r$ contains the facet right of $X$ and $h_l$ contains the facet left of
$X$.
\end{defn}
We can immediately deduce the following property for the stress along $X$.
\begin{lem}\label{lem:stress-orient}
The stress on $X$ is not affected by the order of the elements in $X$.
\end{lem}
\begin{proof}
When reordering the elements of $X$, only the sign of a creasing along $X$ might change.
When it flips then our  notion of left and right face will also be exchanged, which means that
the left face becomes the right face with respect to $X$ and vice versa.
Let  $\bar X$ be a reordering of $X$, such that the creasing along $X$ changes it sign and the right/left
position of the faces is swapped. Suppose that $f$ contains the face left of $\bar X$ and right of $X$,
and that $g$ contains the other face incident to $X$. Then we have
by~Lemma~\ref{lem:crease}
$$\omega_{\bar X}=c(f,g,\bar X)=-c(f,g,X)=c(g,f,X)=\omega_X.$$
\end{proof}

We remark that the stress that we define is an equilibrium stress in the classical sense. In fact,
in three dimensions it corresponds to the stress that is given by the Maxwell--Cremona correspondence.
In particular,
the formulation based on~\eqref{eq:defcrease} can be found in a similar form in 
Hopcroft and Kahn~\cite[Equation~(11)]{HK92}.

%
%
%

\subsection{Convexity}
The difficult part for the height assignment is to choose the heights such that the
resulting lifting gives a \emph{convex} realization of the simplicial complex. 
To guarantee the convexity of the
final embedding we use the stresses induced by the lifting as a certificate. By knowing the signs
of all stresses we can determine if the selected $z$-coordinates produce a
convex embedding with help of the following lemma.
\begin{lem}\label{lem:stresssign}
Let $P$ be the lifting of a flat embedding, such that all vertices
have been assigned with a nonnegative $z$-coordinate. If 
\begin{itemize}
  \item for all $f_X$ incident to $f_B$ we have $\omega_X <0$ and
\item for all other $f_X$ we have $\omega_X>0$,
\end{itemize}
then $P$ is a convex polytope.
\end{lem}
\begin{proof}
As a preliminary step we show that $P$ is \emph{locally convex}. By this we mean that for
every facet $f_T\ne f_B$ of $P$ spanned by $T$, all
adjacent facets lie ``below'' the hyperplane spanned by $T$.
Let $f_S$ be a face that is adjacent to $f_T$ along the ridge $f_X$.
The sequences $S$ and $T$ are ordered such that they coincide on the first $d-2$
elements, that is $X$. We assume for now that $f_T$ is right of $f_X$, which implies $\llbracket T \rrbracket <0$.
\\{\bf Case 1 ($S\not=B$)}: 
We denote by $\er$ the vertex in $S$ that is not in $T$ projected into the $z=0$ hyperplane.
By definition $f_S$ is left of $f_X$ and hence $\llbracket S \rrbracket>0$. 
Since $\omega_X>0$, we
know that $c(f_S,f_T,X)>0$.
This leads to
$$\omega_X>0 \iff  c(f_S,f_T,X)>0 \iff \frac{z_T(\er)-z_S(\er)}{\llbracket S
\rrbracket}>0 \iff z_T(\er)>z_S(\er).$$
Hence the interior of $f_S$ lies below $h(T)$.
\\{\bf Case 2 ($S=B$)}:  
The only difference here is that by assumption $\omega_X<0$ and that $\llbracket
S \rrbracket$ and $\llbracket T \rrbracket$ have the same sign, i.e.,
$\llbracket S \rrbracket<0$. Therefore, 
$$ \omega_X<0 \iff c(f_S,f_T,X)<0 \iff \frac{z_T(\er)-z_S(\er)}{\llbracket
S\rrbracket}<0 \iff z_T(\er)>z_S(\er).$$
Again, the interior of $f_S$ lies below $h(T)$.

For both cases we have assumed that $f_T$ is right of $f_X$. 
If $f_T$ would be left of $f_X$ instead, then
the sign of $\llbracket
S \rrbracket$ would change and furthermore, $\omega_X = c(f_T,f_S,X)= - c(f_S,f_T,X)$ by Lemma~\ref{lem:crease}.
Both effects cancel, i.e., for case~1 we have
$$\omega_X>0 \iff  c(f_S,f_T,X)<0 \iff \frac{z_T(\er)-z_S(\er)}{\llbracket S
\rrbracket}<0 \iff z_T(\er)>z_S(\er),$$
and the same ``cancellation'' happens in case~2.
  
We have shown that $P$ is locally convex and will extend this to global convexity. This follows in 
our setting by Theorem~2.3.20 from the book  of De~Loera, Rambau and Santos~\cite{drs10}. For completeness 
we also give the detailed argument in this place.
In particular, we show that for every facet $f_T\neq f_B$ all other vertices of $P$ lie below $h(T)$. 
We first generalize the observations that proved the local convexity. 
In the above estimations we can pick as $\er$ any vertex that lies in the halfspace
of $\R^{k-1}$ that is bounded by the supporting hyperplane of $\pi(f_X)$ and that does not contain $\pi(T)$.
It still holds that $z_T(\er)> z_S(\er)$.

We are now ready to prove that $P$ is a convex polytope. 
Let $\p$ be a vertex of $P$ not in $T$ with $z$-coordinate $z_p$ and let $\pi(\p)=\er$. 
Consider a segment $\ell$ in $\R^{k-1}$ that connects $\er$ with a point on $\pi(T)$.
We can assume that $\ell$ avoids the vertices of $\pi(P)$ in its interior, otherwise we perturb $\ell$ slightly.
Let the supported hyperplanes of the facets visited when traversing $\ell$ be (in order) $h(T)=h(1),h(2),\ldots h(k)$.
Due to the local convexity we have for $\er = \pi(\p)$ that $z_{h(k)}(\er)<z_{h(k-1)}(\er)$. Furthermore, 
due to the generalized observation in the previous paragraph we have for all $j\le k$ that $z_{h(j)}(\er)<z_{h(j-1)}(\er)$.
This yields
$$ z_T(\er) = z_{h(1)}(\er) > z_{h(2)}(\er) > \cdots > z_{h(k)}(\er) = z_p,$$
which proves the assertion for $f_T\neq f_B$. The base face $f_B$ lies in the $z=0$ hyperplane, which 
is clearly a bounding hyperplane since  all vertices have positive $z$-coordinates. Thus all supporting hyperplanes 
of the facets of $P$ are bounding hyperplanes and therefore $P$ is convex.
\end{proof}

\begin{figure}[hbt]
\begin{center}
\label{fig:globalconvex}
  \includegraphics[width=.65\columnwidth]{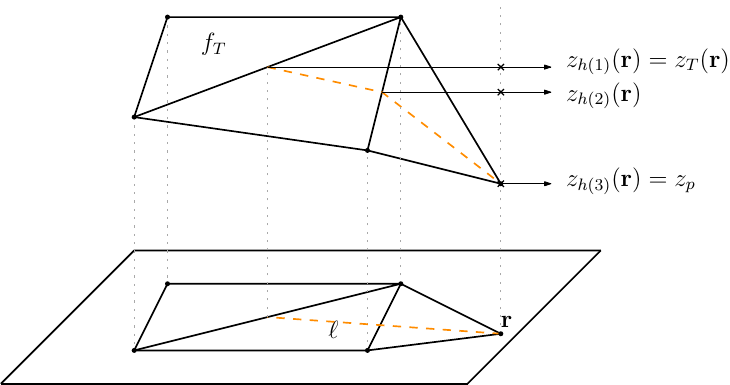}
  \caption{Proof of Lemma~\ref{lem:stresssign}: A cascading sequence of inequalities (coming from the local convexity of $P$) shows that
every vertex of $P$ lies below the supporting hyperplane of a facet $f_T\neq f_B$. 
}
\end{center}
\end{figure}

\subsection{Keeping track of the stresses while stacking}

We discuss next how a stacking operation changes the associated stresses in a
lifted barycentric subdivision. Let $\Delta_D$ be a $(d-1)$-simplex of a
simplicial complex embedded in $\R^d$. We stack the new vertex $\p$ on top of
$D$. To determine the stacking geometrically, we describe the location of
$\er:=\pi(\p)$ inside $\pi(\Delta_D)$ with barycentric coordinates, that is we
specify the absolute areas of the $(d-1)$-simplices containing $\er$ in the projection.
Additionally, we describe how far the $z$-coordinate $z_p$ of $\p$ lies above
$h(D)$. This will be denoted  by $\zeta:=z_p - z_D(\er)$. We refer to $\zeta$ as
the \emph{vertical shift}. 
For convenience we assume that the stacking
process increases the $z$-coordinate of the stacked point, since this will
always be the case in the following. However, the following
observations can be easily generalized for stacking operations that decrease the
$z$-coordinates of the stacked point. 

The addition of $\p$ has two effects. First of all, by stacking $\p$ we create
new ridges, furthermore, the existing ridges at the boundary
$f_D$ are now incident to new facets and hence the corresponding
stress is altered. We refer to the newly introduced ridges as the
\emph{interior ridges}, and to the ridges on the boundary of $f_D$ as
the \emph{exterior ridges} of the stacking operation. 

We discuss first how  the stresses on the exterior ridge are altered. 
\begin{lem}\label{lem:ext}
Assume that we have stacked a new vertex $\p$ 
on a convex polytope realized in $\R^d$ with stress $\omega$ as described in Lemma~\ref{lem:stresssign}. 
Let $X$ be the point set of an exterior ridge  $f_X$ (as defined above).
After the stacking the new stress $\hat \omega$
equals
\[\hat \omega_X:=\omega_X- \frac{\zeta}{\left\lvert \llbracket S \rrbracket\right\rvert}
,\]
where $S=X\conc\p$ and $\zeta$ is the vertical shift of the stacking.
\end{lem}
\begin{proof}
Let  $T$ be the ordered point set, such that $f_T$ and $f_S$ are adjacent after stacking $\p$, and $f_X$
is the intersection of $f_S$ and $f_T$. 
We denote as $h$ the hyperplane that contains the face $f_D$ on which $\p$ was
stacked onto; see Fig.~\ref{fig:ext}. \\
%
Assume for now that $f_S$ is left of $f_X$, which implies that $\llbracket S \rrbracket > 0$.
This gives for $\er=\pi(\p)$
\begin{equation*}\label{equ:ext}
\hat \omega_X=c(f_S,f_T,X)=\frac{z_T(\er)-z_S(\er)}{\llbracket S \rrbracket}=
\frac{z_T(\er)-(z_h(\er)+\zeta)}{\llbracket S
\rrbracket}=\omega_X-\frac{\zeta}{\llbracket S \rrbracket}= \omega_X-\frac{\zeta}{\lvert\llbracket S \rrbracket \rvert}.
\end{equation*}
If $f_S$ is however right of $f_X$ we obtain
\begin{equation*}
\hat \omega_X= - c(f_S,f_T,X)=-\frac{z_T(\er)-z_S(\er)}{\llbracket S \rrbracket}=
-\frac{z_T(\er)-(z_h(\er)+\zeta)}{\llbracket S
\rrbracket}=\omega_X+\frac{\zeta}{\llbracket S \rrbracket}= \omega_X-\frac{\zeta}{\lvert\llbracket S \rrbracket \rvert},
\end{equation*}
since $\llbracket S \rrbracket <0$. 
%
 \end{proof}

\begin{figure}[htb]
\begin{center}
  \includegraphics[width=.52\columnwidth]{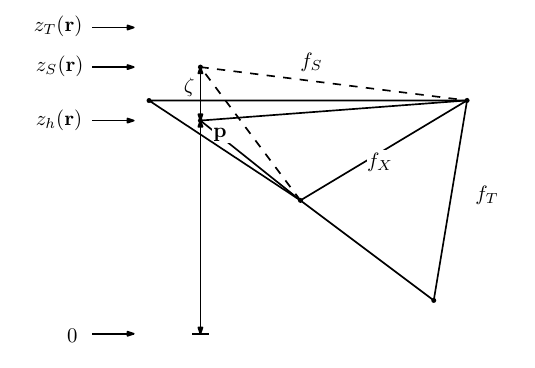} 
\caption{The situation (for $d=3$) as discussed in Lemma~\ref{lem:ext}: the effect of stacking the vertex $\p$ for an exterior ridge.}%
  \label{fig:ext}
\end{center}
\end{figure} 

\begin{lem}\label{lem:int}
Assume we have stacked a vertex $\p$ on top of some  face $f_D$. 
Let $X$ be the point set of an interior ridge $f_X$ as defined above.
Furthermore, let $S:=X\conc \es_d$ and $T:=X \conc \te_d$ be the sequence of
vertices of the faces separated by $X$. After the stacking the stress on $X$
equals the positive number 
\[\omega_X:=\zeta  \left \lvert  \frac{ \llbracket D \rrbracket }{\llbracket S
\rrbracket \llbracket T \rrbracket} \right \rvert.\] 
\end{lem}
\begin{proof}
Fig.~\ref{fig:int} depicts the situation described in the lemma.
The facet $f_S$ might be either left or right of $f_X$.
Due to Equation~\eqref{eq:altstress} we have
$$ \omega_X= \pm c(f_S,f_T,X)=  \pm \frac{[T\conc \es_d]}{\llbracket T
\rrbracket\llbracket S \rrbracket}.$$
Assume we have doubled $f_D$ and stack with an $\varepsilon$-small $\zeta$ at
the ``top side'' of the induced complex. This will surely generate a convex
$d$-simplex, and by the computations in the proof of 
 Lemma~\ref{lem:stresssign} the sign of the induced $\omega_X$ is positive. Clearly, the sign of $\omega_X$ does not change if we increase the vertical shift.
 
We know that $\lvert [T\conc \es_d] \rvert$ is $d!$ times the
volume of the corresponding simplex. This volume on the other hand can be
expressed as $\zeta \cdot \lvert \mathit{vol}(\pi(\Delta_D)) \rvert/d$, and $\lvert \llbracket D \rrbracket \rvert$ 
equals $(d-1)!$ times the volume of $\pi(\Delta_D)$. Hence
we have $\lvert [T\conc \es_d] \rvert= \zeta \cdot \lvert \llbracket D
\rrbracket \rvert$ and the statement of the lemma follows.  
\end{proof}
\begin{figure}[htb]
\begin{center}
  \includegraphics[width=.47\columnwidth]{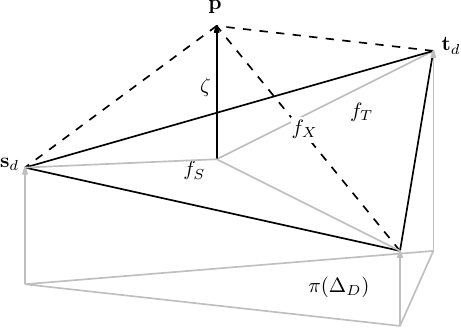} 
\caption{The situation (for $d=3$) for an interior ridge as discussed in
Lemma~\ref{lem:int}.}
  \label{fig:int}
\end{center}
\end{figure}

\subsection{The tree-representation of a stacked polytope}
\label{subsec:treerep}

The combinatorial structure of a polytope is typically encoded by its face lattice.
If the polytope is a stacked polytope, we can also describe its structure by
``recording'' in which way we 
have carried out the stacking operations. The most natural way to keep track of
the stacking operations is an ordered rooted tree. Let $P$ be a stacked polytope, then
the \emph{tree-representation} of $P$ is a $d$-ary tree $\mathcal{T}(P)$, which is
defined as follows: The leaves of the trees are in one-to-one correspondence to
the facets of $P$, with the exception of the face $f_B$, which is not present in the
tree. Interior nodes correspond to $d$-cliques in the graph of $P$ that used to be
facets at some point during the stacking process.\footnote{
When considering an intermediate configuration in the 
stacking process, we refer to the $d$-cliques of $P$ that are faces in the intermediate 
polytope also as \emph{facets}, although they are not necessarily facets of $P$.}
 The root represents the initial
copy of $f_B$ in the beginning. The children of a node $v$ represent the faces
that were introduced by stacking a vertex onto the face associated with $v$. To 
reproduce the combinatorial structure from $\mathcal{T}(P)$ we 
fix an ordering for the edges emanating from an interior node $v$, such that it
is possible to map the children of $v$ to the faces generated by  stacking onto
$v$'s face in a unique way. By this the
mapping between the faces of~$P$ and the nodes of~$\mathcal{T}(P)$ can be reproduced. If $P$ has~$n$ vertices, then $\mathcal{T}(P)$ has $(n-d)(d-1)+1$ leaves and $n-d$ interior nodes.

We will use the tree-representation of $P$ to specify the geometry of the flat embedding of
$P$. This can be achieved by assigning for every leaf $v$ of $\mathcal{T}(P)$ a
rational number (weight), which corresponds to the volume of $v$'s face in the flat
embedding of $P$. More precisely, the weight specifies the volume of $v$'s face
times $(d-1)!$. To emphasize this relationship, we call the weights \emph{face-weights}.
We extend the \fw assignments for the interior nodes by
summing for every node the \fws of its children recursively up. After we have determined
the location of $f_B$, such that its volume is $(d-1)!$ times the \fw of the root of
$\mathcal{T}(P)$, the coordinates of the whole flat embedding are determined.
In particular, for every stacking operation, the (normalized) \fws specify the
location of the stacked vertex since they denote its barycentric coordinates.
Thus, by traversing $\mathcal{T}(P)$ we can produce the flat embedding
incrementally in a unique way.

\section{The embedding algorithm}
\label{sec:algo}
We assume that the combinatorial structure of the 
stacked polytope is given in form of its tree-representation
$\mathcal{T}(P)$. The embedding algorithm works in three steps. First we
generate the \fws for $\mathcal{T}(P)$ and fix the coordinates for
the face $f_B$. This will give us a flat embedding of the polytope. In the next
step we ``lift'' the polytope, by defining for every vertex $v_i$ the
\emph{vertical shift} $\zeta_i$. Assume that we have stacked $v_i$ onto the
face $f$. By
construction we pick always positive vertical shifts. By carefully choosing the
right vertical shifts we obtain an embedding of $P$ as a \emph{convex} polytope, 
however this embedding
does not fit on a polynomial grid. In the final step we round the 
points to appropriate grid points, while maintaining convexity.

\subsection{Balancing \fws}
We apply a technique from data structure analysis called the \emph{heavy
path decomposition} (see Tarjan~\cite{T83}). Roughly
speaking, it decomposes a tree into paths, such that the induced
hierarchical structure of the decomposition is balanced. We continue with a brief
review of the heavy
path decomposition.
Let $u$ be a non-leaf of a rooted tree $T$ with root $r$ ($r=u$ is possible). We denote by $T_u$ the
subtree of $T$ rooted at $u$. Let $v$ be the child of $u$ such that $T_v$ has
the largest number of nodes (compared to the subtrees of the other children of
$u$), breaking ties arbitrarily. We call the edge $(u,v)$ a \emph{heavy edge},
and the edges to the other children of $u$ \emph{light edges}. The heavy edges
induce a decomposition of $T$ into paths, called \emph{heavy paths},
and light edges; see Fig.~\ref{fig:hetd}. 
The node on a heavy path that is closest to the root is called its \emph{top node}.
We call a heavy path with its
incident light edges (ignoring the possible edge from its top node to its parent) a \emph{heavy caterpillar}. The heavy
path decomposition decomposes the
edges of $T$ into heavy caterpillars. We say that two heavy caterpillars are
adjacent if their graphs would be adjacent subgraphs in $T$. This
adjacency relation induces a hierarchy,  which we represent as a rooted tree
$\mathcal{H}(T)$. The nodes in $\mathcal{H}(T)$ are the heavy caterpillars, and its edges
represent the adjacency relation between caterpillars. The root of 
$\mathcal{H}(T)$ is the heavy caterpillar that includes the root of $T$. 
When $(u,v)$ is a light edge and $u$ is the parent of $v$ then the size of $T_v$
is at most half as big as the size of $T_u$.
Hence, every root-leaf path in $T$ can visit at most $\log t$
 light edges for $t$ being the size of~$T$. Fig.~\ref{fig:hetd} shows an
example of a tree-representation and its associated hierarchy as a tree.
\begin{figure}[htb]
\begin{center}
\begin{tabular}{cp{.1cm}c}
  \includegraphics[width=.46\columnwidth]{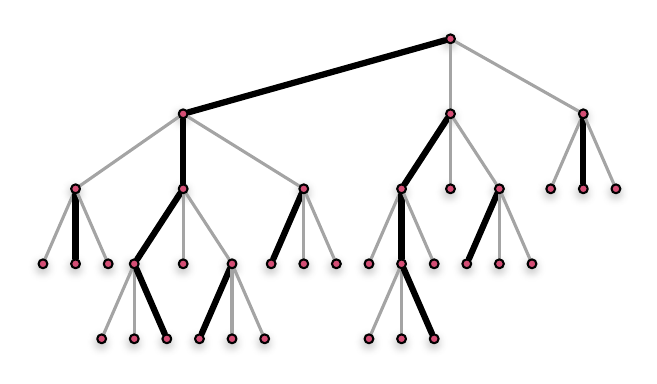} & &
  \includegraphics[width=.46\columnwidth]{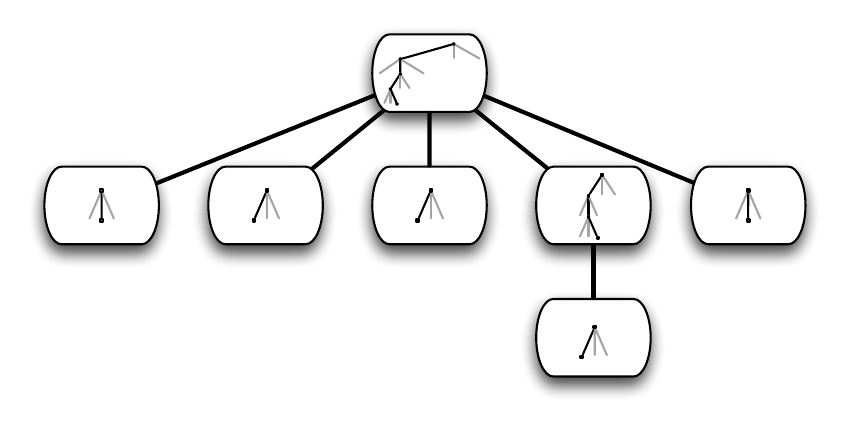} 
  \\
\small{(a)} & & \small{(b)} 
\end{tabular}
\caption{A tree-representation $\mathcal{T}$ of a stacked 3-polytope (a) and the
corresponding hierarchy based on the caterpillars as tree
$\mathcal{H}(\mathcal{T})$ (b).}
  \label{fig:hetd}
\end{center}
\end{figure} 

The assignment of the \fws is guided by the heavy
path decomposition of
$\mathcal{T}(P)$. We call a node $v$ \emph{balanced} if (1) the \fws of $v$'s light-edge children 
are all the same and (2) the \fw of the heavy edge child is not smaller than the \fw of every
light edge child. If every node is balanced we call the \fw assignment balanced. An example of
a balanced tree-representation is depicted in Fig.~\ref{fig:balancedexample}.

\begin{lem}\label{lem:balance}
We can find a balanced set of integer \fws for $\mathcal{T}(P)$ such that the
\fw associated with the root is at most $dn^{\log 2d}$ and no \fw
is less than 1. Here, $n$ denotes the number of vertices in the corresponding graph of the stacked polytope.
\end{lem}
\begin{proof}
Let $H=\mathcal{H}(T)$ be the tree representing the hierarchy of the
heavy caterpillars for some $d$-ary tree~$T$. The height of~$H$ is the length (number of edges) of its longest root-leaf path.
For~$T$ we call the height of $\mathcal{H}(T)$ its $\emph{hpd-height}$.

To prove the lemma we use the following
claim. Since the proof of the claim will be constructive it will also induce a strategy of how to set the \fws.
\begin{description}
\item[\bf Claim:] For all $k\ge 0$ the following holds: If $T$ is a $d$-ary tree with $m$ leaves and hpd-height~$k$ then we can find
a set of balanced \fws, such that the  \fw of the root is at most $m d^k$.
\end{description}

We prove the claim by induction on~$k$.
A tree with hpd-height~0 is a heavy caterpillar, whose subtrees on the light edges are leaves.
In this case it suffices to set the \fws of all leaves to~1 and then propagate the \fws to the interior nodes.
Clearly, this \fw assignment is balanced and the \fw of the root equals the number of leaves. 
Thus, the claim holds for this (base) case.

We continue with the induction step.
Assume that the claim holds for all trees with hpd-height less then $k$. We can now balance the \fws of
any tree of hpd-depth $k$ as follows (see also Fig.~\ref{fig:balance} for an illustration):
 Let the heavy path incident to the root be $h$. The subtrees connected to $h$
via a light edge have all depth less than $k$. For these subtrees we use our induction hypothesis and (recursively)
determine the corresponding face weights and propagate these weights to $h$. If a subtree is a single leaf it will be assigned with a \fw of~1.
We are left with balancing the nodes of $h$. So let $v$ be a node on $h$
and let $u^+_v$ be one of its light edge children with the largest \fw. We will now make each light child $u\neq u_v^+ $ as ``heavy''
as $u_v^+$. To do so, we determine the difference $\delta$ between the \fw of $u^+_v$ and the \fw of~$u$. Next, we increase the \fw of every node on the heavy path with top node $u$ by $\delta$.
Notice that this keeps the \fws in the subtree rooted at $u$ balanced. To make the updates consistent we also increase
the \fw on $h$ from $v$ to its root by $\delta$.
We continue with the other children of $u$ in the same fashion and then repeat this for every node on $h$.
We also have to guarantee that 
the \fw of every light edge child is not larger than the \fw of its heavy edge child. To ensure this, we
 increase the \fw of the leaf of $h$ by $\delta_r$ and propagate this weight along~$h$, where $\delta_r$ is the largest \fw of one of the 
light edge subtrees hanging off of $h$. Note that this is done only once for the heavy path~$h$.
 By this we have obtained a balanced set of \fws for $T$. 

We are left with bounding the new \fw of the root. Let $S$ be the sum of the \fws of all the light edge subtrees
of $h$ before balancing the nodes of $h$, which is also the \fw of the root minus 1 at this time. 
By the induction hypothesis we have $S\le md^{k-1}$.
When balancing the light subtrees on a vertex $v$ we increased the \fw of the root by some $\delta$
at most $d-2$ times. The ``charge'' $\delta$ was not larger than the \fw of the corresponding $u_v^+$. 
The total increment at this stage is therefore less than $(d-2)S$. The final increment by $\delta_r$ is at most $S-1$.
Therefore, we have that the \fw of the root is at most 
$$(S+1)+(d-2)S+(S-1)=dS\le d m^k$$ 
and the proof of the claim follows.

The statement of the lemma follows now from the claim. According to Subsec.~\ref{subsec:treerep} the tree $\mathcal{T}(P)$ has less than $(d-1)n$ leaves. When traversing a light edge $(u,v)$ the number of interior 
nodes is at least halved. The tree~$\mathcal{T}(P)$ has less than $n$ interior nodes. Thus, every path in~$\mathcal{T}(P)$ from the top node of a heavy path to the root of~$\mathcal{T}(P)$ visits at most $\log n$
light edges. As a consequence, the maximal hpd-height  
is bounded by $ \log n$. The 
total \fw of the root is therefore at most $(d-1)n \cdot
 d^{\log n} < 
 dn \cdot
 n^{\log d} = 
 d n^{\log 2d}$. 
\end{proof}

%
%
%

\begin{figure}[htb]
\begin{center}
\begin{tabular}{cp{1cm}c}
  \includegraphics[scale=.7]{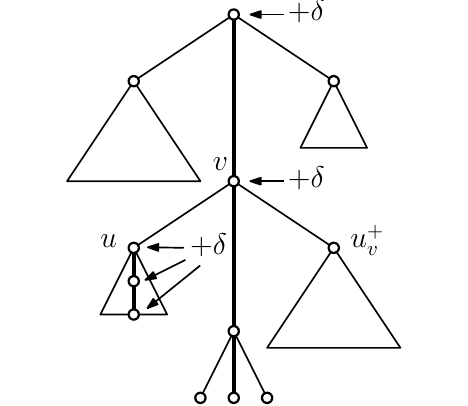} & &
    \includegraphics[scale=.7]{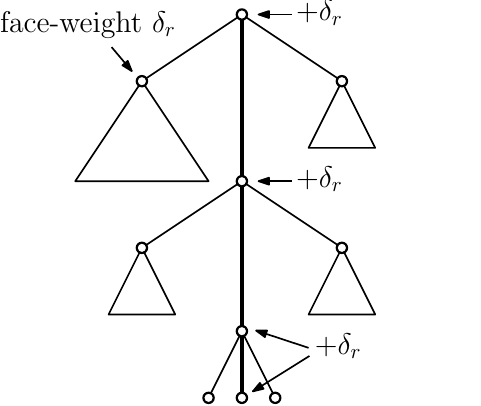} 
  \\
\small{(a)} & & \small{(b)} 
\end{tabular}
\caption{Modifications of the \fw as done in the proof of Lemma~\ref{lem:balance}. (a) Sketch how to make all light children have the same \fw. (b) Sketch how to 
keep the heavy edge child the child with largest \fw.
}
  \label{fig:balance}
\end{center}
\end{figure} 
In the following we denote
the \fw of the root as $R$. 
To finish the definition of the flat embedding we fix the shape of $f_B$ as
follows: One of the vertices of $B$ lies at the origin. We set
$L=\sqrt[d-1]{ R}$ and place the  other vertices of $B$ at 
$L\cdot \mathbf{e}_i$, for $ \mathbf{e}_i$ being a vector of the standard basis
of $\mathbf{R}^{d-1}$, such that $B$ spans the $(d-1)$-simplex $\Delta_B$. By
construction the volume of the simplex equals $1/(d-1)!$ times the \fw of the
root of $\mathcal{T}(P)$. Without loss of generality we also assume for the remainder 
that $\llbracket B \rrbracket >0$.
It follows  that $$\llbracket B \rrbracket =R\le dn^{\log (2d)}. 
$$ 
By our choice of $f_B$ the volume of 
every face in the flat embedding equals $(d-1)!$ times  the volume of its \fw.
\begin{figure}[htb]
\begin{center}
  \includegraphics[width=.6\columnwidth]{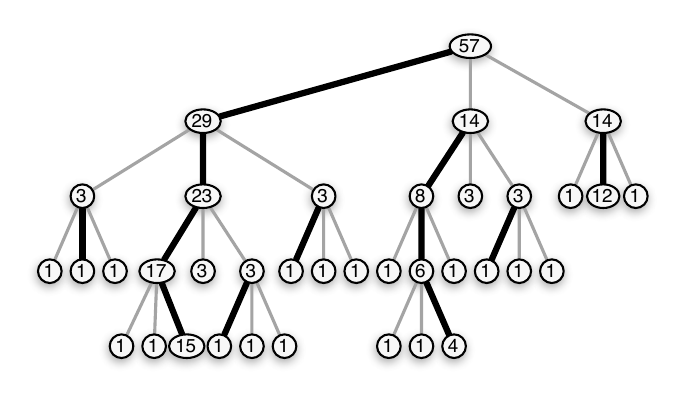} 
\caption{The tree-representation of Fig.~\ref{fig:hetd} with balanced \fws, denoted in the nodes.
}
  \label{fig:balancedexample}
\end{center}
\end{figure}

\subsection{Assigning heights}
After we have assigned the \fws we specify the heights of the lifting by determining
the vertical shift $\zeta_i$ for every vertex $v_i$ based on the 
 \fws. Let us assume vertex $v_i$ was stacked as $\p$
 onto some face $f_D$. Let the boundary of $f_D$ be formed by the 
 ridges $f_{X_1},\ldots, f_{X_d}$. The stacking introduces the new facets
$f_{Y_1},\ldots, f_{Y_d}$, where for $1\le j \le d$ we have $Y_j:=X_j\conc \p$. The
nodes of the faces $f_{Y_j}$  in
$\mathcal{T}(P)$ that are connected to the node of $f_D$ via a light edge
have all the same \fw. Let this
weight be $B_i$. The only (possible) other \fw for a node 
associated with one of the faces $f_{Y_j}$ is denoted
by $A_i$. 
We set as the vertical shift for the vertex $v_i$
\begin{equation}\label{eq:zetadef}
\zeta_i:=A_i \cdot B_i.
\end{equation}
Note that $A_i \ge B_i$, since the \fws were balanced. 


\begin{lem}\label{lem:stressvalues}
The embedding induced by the \fws and the vertical shifts $\zeta_i$
guarantees:
\begin{itemize}
\item[1.)] For every ridge $f_X$ not on $f_B$ we have $\omega_X\ge1$.
\item[2.)] For every ridge $f_X$ on $f_B$ we have $0>\omega_X\ge -R\ge -d\cdot
n^{\log (2d)}.$
\end{itemize}
\end{lem}

\begin{proof}
We first bound the stresses on faces not on $f_B$.
Recall that by construction the \fw of a facet $f_D$ coincides with
$|\llbracket D \rrbracket |$.
We study how the stress on $\omega_X$ evolves during the stacking process. The 
initial value of $\omega_X$ is assigned by some stacking operation
that introduced $f_X$.
We assume that this stacking operation stacked a vertex on the face $f_D$. 
The \fws of the new facets introduced
by the stacking are all the same, namely  $B_i$, except for one possible larger 
\fw, namely $A_i$. Let $f_S$ and $f_T$ be the two facets incident to $f_X$ such
that $| \llbracket S\rrbracket | \ge | \llbracket T\rrbracket | $.
By Lemma~\ref{lem:int}, we have that at this moment
$$\omega_X = \zeta_i \left\lvert  \frac{\llbracket D \rrbracket} {\llbracket S\rrbracket  \llbracket T\rrbracket}   \right\rvert 
= A_i B_i \left\lvert  \frac{\llbracket D \rrbracket} {\llbracket S\rrbracket \llbracket T\rrbracket}   \right\rvert  \ge \left\lvert \llbracket D \rrbracket \right\rvert, $$
since $|\llbracket S \rrbracket| \le A_i$ and $| \llbracket T \rrbracket |= B_i$.

This positive initial stress decreases when stacking on a facet that has 
$f_X$ on the boundary. So assume we stacked $\p_k$ on such a facet. Let $C_X(k)$ denote the amount of the decrement  due to this stacking.  
By Lemma~\ref{lem:ext} we have  $C_X(k)=\zeta_k/ \lvert \llbracket S_k \rrbracket \rvert$, 
where $S_k=X\conc \p_k$. Recall that $\zeta_k=A_kB_k$, where $A_k$ and $B_k$ are the two different values of
\fws of faces introduced by stacking $\p_k$. It follows that $|\llbracket S_k \rrbracket|\in\{A_k,B_k\}$, 
and therefore $C_X(k)=\{A_k,B_k\} \setminus \{ |\llbracket S \rrbracket |\}$ equals either $A_k$ or $B_k$, and hence, is an integer. 
We charge the value of $C_X(k)$ to a  \fw of a face $f_{Y_k} \neq f_{S_k}$ that was introduced when stacking 
$\p_k$. Stacking operations that
decrease $\omega_X$ further, stack either onto $f_{S_k}$, or onto the opposite facet incident to $f_X$ (see Fig.~\ref{fig:propagation}).
As a consequence, the projected 
facets $\{\pi(f_{Y_j}) \mid \text{stacking of  $\p_j$ decrements $\omega_X$} \}$ have disjoint interiors and are properly contained inside $\pi(\Delta_D)$.
Since it is
not possible to cover $\pi(\Delta_D)$ with these faces completely, 
the difference $\left\lvert \llbracket D \rrbracket \right\rvert- \sum_j C_X(j)$ is at least 1 
and therefore $\omega_X\ge \left\lvert \llbracket D \rrbracket \right\rvert- \sum_j C_X(j)\ge 1$. 
\begin{figure}[htb]
\begin{center}

  \includegraphics[width=.45\columnwidth]{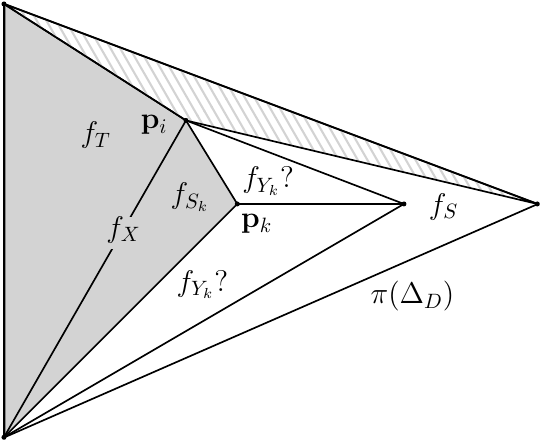} 
\caption{When stacking $\p_k$ the stress on $X$ is decreased by at most the \fw charged to $f_{Y_k}$. 
There are two candidates for $f_{Y_k}$. All further negative charges will be attributed to ``volumes'' contained 
inside $S_k$ and $T$. The face with the crossed-line pattern will never be charged to $\omega_X$.}
  \label{fig:propagation}
\end{center}
\end{figure}

For the stresses on the boundary we can argue as follows. The stress on $X$ 
is  a combination of several negative charges, but without 
having an initial positive charge. By the above argument, the negative charges
 are attributed to faces with disjoint interiors. All these 
faces are contained inside $\pi(\Delta_B)$ and hence $|\omega_X|$ is at 
most $\llbracket B \rrbracket $. By Lemma~\ref{lem:balance} 
we have that $\llbracket B \rrbracket  =R \le 
d \cdot n^{\log (2d)}, $ and the lemma follows.
\end{proof}

\subsection{Rounding to grid points}
Let us wrap up what we have constructed so far. Based on the tree-representation we
have defined weights for each face. This gave rise to a realization of a 
projection of the stacked polytope. As a next step we have computed a lifting based 
on the \fws. This construction produces a convex realization of the desired stacked
polytope, we even know that the stresses are polynomially related. However, the realization
does not lie on a polynomial grid yet. To obtain an integer realization we round the coordinates
of the points down. The rounding will be carried out in two steps. First we 
perturb all coordinates such that they are multiples of some parameter $\alpha$, resp. $\alpha_z$ for the $z$-coordinates.
In a second step we scale the perturbed embedding by multiplying all coordinates with $1/\alpha$, resp. with
$1/\alpha_z$, for the $z$-coordinates. The details for the rounding procedure are little bit more subtle: We start
with rounding the coordinates in the flat embedding, then we update the vertical shifts $\zeta_i$ slightly.
The $z$-coordinates are rounded with respect to the lifting defined by the modified vertical shifts.

Since projected volumes play an important role
 in our construction we discuss as a first step how the rounding will effect these volumes.
\begin{lem}\label{lem:perturbation}
If we round the coordinates of the flat embedding down, such that  every coordinate is a multiple of $\alpha$
  we have that for every facet $f_X$
$$ \left\lvert \llbracket X \rrbracket - \llbracket X' \rrbracket\right\rvert  \le \alpha d^2   L^{d-2} , $$
where $X'$ denotes the points $X$ after the rounding.
\end{lem}
\begin{proof}
Remember that the points $X$ are contained inside the simplex $\Delta_B$, which is spanned by the standard basis vectors of $\R^{d-1}$ scaled by $L=\sqrt[d-1]{R}$
and the origin. 
We first show how the rounding of the first coordinate for every $x \in X$ effects $\llbracket X
 \rrbracket$. The point set $X$ after rounding the first coordinate is named $X^1$.
Let $E(1)$ denote the change of $\llbracket X \rrbracket$ in the worst case, that is,  $E(1)=\max_{X} \left\lvert \llbracket X \rrbracket - \llbracket X^1 \rrbracket\right\rvert $. We denote by $x_i\in \R$ the first coordinate (the $x$-coordinates) of the $i$th point in  $X$, and
 for a point sequence  $X$ we denote by  $X_{-i}$ the set with removed $x$-coordinates and without the $i$th point of $X$. Note that
$X_{-i} = X^1_{-i}$. Let the change of the coordinates $x_i$ be $\varepsilon_{i} \le \alpha$.
We estimate $E(1)$ using the Laplace expansion for the determinants (along the row for the $x$-coordinates) by
\begin{align*}
E(1)=\left\lvert \llbracket X \rrbracket - \llbracket X^1 \rrbracket\right\rvert & 
\le \left\lvert \sum_{i=1}^d (-1)^i x_i \llbracket  X_{-i} \rrbracket  -  \sum_{i=1}^d (-1)^i (x_i+\varepsilon_{i}) \llbracket  X^1_{-i} \rrbracket  \right\rvert \\
& \le \left\lvert \sum_{i=1}^d (-1)^i x_i \llbracket  X_{-i} \rrbracket  -  \sum_{i=1}^d (-1)^i (x_i+\varepsilon_{i}) \llbracket  X_{-i} \rrbracket \right\rvert \\
&  \le \sum_{i=1}^d \left \lvert \alpha \llbracket  X_{-i} \rrbracket \right\rvert  \\
& \le d \alpha L^{d-2} .
\end{align*}
For the last transition we upper bounded the determinants $\llbracket  X_{-i} \rrbracket$ by $L^{d-2}$. 
This bound follows, since the maximal volume of a $(d-2)$-simplex  in $\Delta_B$ is spanned by the $d-1$ points 
of $B\setminus \mathbf{0}$.  This point set would also maximize $\llbracket  X_{-i} \rrbracket$
and hence this determinant is at most $L^{d-2}$. 

Note that there was nothing special in rounding the \emph{first} coordinate (while keeping the others fixed). If we would
round only the $k$th coordinate, we obtain the same estimation. Thus, we can apply the rounding
of a single coordinate one by one for each of the $(d-1)$-coordinates of $\R^{d-1}$. Every time
we round in one coordinate  we
introduce an additive ``error'' of  $d \alpha L^{d-2}$  to $\left\lvert \llbracket X \rrbracket - \llbracket X' \rrbracket\right\rvert$.
Thus in total we get the asserted bound of $\left\lvert \llbracket X \rrbracket - \llbracket X' \rrbracket\right\rvert  \le \alpha d^2   L^{d-2}$.
\end{proof}

In the following we use $\llbracket X' \rrbracket$ and similar expressions to denote the corresponding determinants after rounding.  
\begin{cor}\label{cor:perturbation}
If we round the coordinates of the flat embedding down, such that  every coordinate is a multiple of $\alpha$, we have 
for every facet $f_X$ in the flat embedding 
$$ 1-  \alpha d^2 L^{d-2} \le \frac{\llbracket X' \rrbracket }{\llbracket X \rrbracket } \le 1+ \alpha d^2 L^{d-2}.
$$
\end{cor}
\begin{proof}
The statement follows from Lemma~\ref{lem:perturbation} and the observation that for any facet $f_X$ 
in the flat embedding we have $1\le  | \llbracket X \rrbracket | $.
\end{proof}

As a next step we discuss, how to set the parameter $\alpha$, such that the perturbed flat embedding will still have a
positive stress. To describe the lifting (stress) we defined for every vertex in the flat embedding a vertical shift $\zeta_i$, as given in \eqref{eq:zetadef}.
The definition of $\zeta_i$ was based on the \fws obtained from the balanced tree-representation. We adjust the vertical shifts after the perturbation
slightly. When stacking $\p_i$ we introduced $d$ new faces. Let the faces $f_{A_i}$ and $f_{B_i}$ be the two faces out of the $d$ newly introduced faces with the largest volume.
We define 
$$\zeta_i':=|\llbracket A_i' \rrbracket \llbracket B_i'\rrbracket |.$$ 

\begin{lem}\label{lem:stressperturbation}
When we pick as the perturbation parameter $\alpha =1/(10 d^2  L^{d-2} R)$ then the perturbed
flat embedding with the vertical shifts $\zeta_i'$ induces an embedding, whose interior 
stresses are at least $4/5$, and whose boundary stresses are negative and larger than $-2 R$.
\end{lem}
\begin{proof}
We mimic the strategy of the proof of Lemma~\ref{lem:stressvalues}. Again, all faces in this proof
are considered as projected into the $z=0$ hyperplane. For the proof of the lemma the sign of 
the determinants $\llbracket \cdot \rrbracket$ is not important. For the sake of a simple 
presentation we misuse notation and 
simply write $\llbracket \cdot \rrbracket$ instead of $| \llbracket \cdot \rrbracket |$
in this proof.

Every stress $\omega_{X}$ is a combination 
of a positive stress and several negative stresses attributed to different stacking operations. 
Let $\omega_{X}^+$ be the positive stress, that is the initial nonzero 
stress with respect to the stacking sequence, and let $\omega_{X}^-$ the absolute value of the sum of all negative stresses,
such that $\omega_{X}=\omega_{X}^+ - \omega_{X}^-$.  
To bound $\omega_{X}$ we derive bounds for $\omega_{X}^+$ and $\omega_{X}^-$. We start with the bound on the positive 
stress. Assume $\omega_X^+$ was introduced by stacking $\p_i$ at some face $f_D$, such that due 
to Lemma~\ref{lem:int} we have 
$$\omega_X^+ = \zeta'_i \left\lvert  \frac{\llbracket D' \rrbracket} {\llbracket S'\rrbracket \llbracket T' \rrbracket} \right\rvert. $$
Here, $f_S$ and $f_T$ are the two faces introduced by stacking $\p_i$ that contain $f_X$. The height $\zeta_i'$ is defined as the product
of two \fws corresponding to $f_{A_i}$ and $f_{B_i}$. Assume that $\llbracket A_i' \rrbracket \ge   \llbracket B_i' \rrbracket$ and 
that $\llbracket S' \rrbracket \ge   \llbracket T' \rrbracket$. By the definition of $\zeta'_i$ we have that 
$\llbracket A_i' \rrbracket \ge   \llbracket S' \rrbracket$ and $\llbracket B_i' \rrbracket \ge   \llbracket T' \rrbracket$. 
Therefore, $\omega_X\ge  \llbracket D' \rrbracket$.

The value of $\omega_X^-$ is composed of several ``charges''. Whenever we stack inside a face that contains $f_X$
we increase $\omega_X^-$. Let us study now one of these situations. Assume we stack $\p_k$ inside a face
that contains $f_X$. Let $f_{S_k}$ be the new face that contains $f_X$. By Lemma~\ref{lem:ext} we increase
$\omega_X^-$ by $|\zeta_j'/\llbracket S_k' \rrbracket |:=\mathit{inc}_{X,k}$.
Due to the balanced \fws we had in the unperturbed setting only two different new ``face volume values'' when 
stacking a vertex. Hence for $\zeta_i':=|\llbracket A'_k\rrbracket \llbracket B'_k \rrbracket|$, we had either 
$\llbracket S_k \rrbracket = \llbracket A_k \rrbracket$, or $\llbracket S_k \rrbracket = \llbracket B_k \rrbracket$. 
We define
$$ C_k := \begin{cases}
	B_k & \text{if $\llbracket S_k \rrbracket  =  \llbracket A_k \rrbracket  $} \\
    A_k & \text{if $\llbracket S_k \rrbracket  =  \llbracket B_k \rrbracket  $} 
\end{cases}
$$
and set $\delta_+:=1+ \alpha d^2 L^{d-2}$ and $\delta_-=1- \alpha d^2 L^{d-2}$.
If $\llbracket S_k \rrbracket = \llbracket A_k \rrbracket$ 
we have according to Corollary~\ref{cor:perturbation}
$$\mathit{inc}_{X,k}=\left\lvert \frac{\llbracket A'_k \rrbracket}{\llbracket S_k' \rrbracket} 
\llbracket B'_k \rrbracket \right\rvert \le  \left\lvert \frac{\delta_+\llbracket A_k \rrbracket}{\delta_-\llbracket S_k \rrbracket}  \llbracket B'_k \rrbracket \right \rvert 
= \frac{\delta_+}{\delta_-}  \llbracket B'_k \rrbracket = \frac{\delta_+}{\delta_-}  \llbracket C'_k \rrbracket.$$
The remaining case $\llbracket S_k \rrbracket = \llbracket B_k \rrbracket$ is completely symmetric and does
also give $\mathit{inc}_{X,k}  \le  \frac{\delta_+}{\delta_-}  \llbracket C'_k \rrbracket$.
Let $K$ be the set of vertex indices, whose stacking contributed to $\omega_X^-$. As noticed in 
Lemma~\ref{lem:stressvalues} (see also Fig.~\ref{fig:propagation}), for any two distinct $s,t\in K$ we have that $f_{C_s}$ and $f_{C_t}$
have disjoint interiors, and furthermore the set $\bigcup_{k\in K} f_{C_k}$ is contained in  the perturbed simplex $\pi(\Delta_D)$, but 
``misses'' at least one face. By Corollary~\ref{cor:perturbation} the \fw of the
missing face is at least $\delta_-$. Therefore,
$$\omega_X^-=\sum_{k\in K} \mathit{inc}_{X,k} \le \sum_{k\in K} \frac{\delta_+}{\delta_-}   \llbracket C'_k \rrbracket
=\frac{\delta_+}{\delta_-}  \sum_{k\in K}   \llbracket C'_k \rrbracket \le \frac{\delta_+}{\delta_-}   \left(  \llbracket D' \rrbracket - \delta_- \right)
= \frac{\delta_+}{\delta_-}  \llbracket D' \rrbracket - \delta_+.
$$

We finish the proof by combining the bounds for $\omega_X^+$ and $\omega_X^-$. When we pick $\alpha = 1/(10 d^2 L^{d-1} R)$
as specified in the lemma then we get $\delta_+=1+\frac{1}{10R}$ and $\delta_-= 1-\frac{1}{10R}$.
We can now obtain the following bound for $\omega_X$ when $X$ is an interior ridge (note that $\pi(\Delta_{D'}) \subseteq \pi(\Delta_{B})$) :
\begin{align*}
\omega_X& \ge \omega_X^+- \omega_X^- \\
& \ge \llbracket D' \rrbracket - \left( \frac{\delta_+}{\delta_-}  \llbracket D' \rrbracket - \delta_+\right) \\
& = \delta_+  - \llbracket D' \rrbracket \left( \frac{\delta_+}{\delta_-}-1 \right) \\
& \ge \delta_+ -  \llbracket B\rrbracket \left( \frac{\delta_+}{\delta_-}-1 \right)  = \delta_+  - \delta_+ R  \left( \frac{\delta_+}{\delta_-} -1 \right) \\
&= \frac{80R^2-2R-1}{100R^2-10R}\\
&\ge \frac{4}{5}
\end{align*}
If $X$ is a boundary face  we have $\omega_X^+=0$. Notice that for our choice of $\alpha$ we have  $\delta_+/\delta_- \le 2$, 
for $R\ge 3$. 
We conclude that for a boundary face we have
$$ \omega_X \ge -\left( \frac{\delta_+}{\delta_-}  \llbracket B' \rrbracket  + \delta_+ \right) \ge - \frac{\delta_+}{\delta_-}  \llbracket B' \rrbracket \ge -2  \llbracket B' \rrbracket \ge -2 R.$$
\end{proof}

The choice of the parameter $\alpha = 1 / (10 d^2  L^{d-2} R)$ ensures, that
no volume of a face will flip its sign. In particular, by Lemma~\ref{lem:perturbation}, the change of volume is less 
than $1/(10\llbracket B \rrbracket)$.

To construct  an integer realization we need to round the $z$-coordinates as well. 
We round the $z$-coordinates such that every $z$-coordinate is a multiple of $\alpha_z$, where $\alpha_z$ is some value to be determined
later. The final analysis requires an upper bound for the maximal $z$-coordinate before rounding, which we give in the following lemma. 

\begin{lem}\label{lem:zmax}
For $\zmax$ being the  maximal $z$-coordinate in the lifting of the perturbed flat embedding induced 
by the vertical shifts $\zeta'_i$ we have that
$$0< \zmax < 2 R^2.$$
\end{lem} 
\begin{proof}
The $z$-coordinates of all vertices not on $f_B$ are positive since all vertical shifts are positive. 
We consider the perturbed flat embedding.
Take any boundary ridge $f_X$ of $f_B$. By Lemma~\ref{lem:stressperturbation} we have that $-\omega_X<2 R$.
Let $\p_i$ be the vertex with the highest $z$-coordinate $\zmax$ and set $\er=\pi(\p_i)$. Now
take the facet $f_Y$ adjacent to $f_B$ via $f_X$ (we let $Y$ and $B$ coincide on the first $d-2$ vertices).
 Due to the convexity of the lifting 
$f_Y$ supports a bounding hyperplane, and therefore $z_Y(\er)\ge \zmax$. 
By Equation~\eqref{eq:defcrease} we have $ \omega_X\llbracket Y \rrbracket = z_B(\er)-z_Y(\er)=- z_Y(\er).$
Since $f_Y$ is properly contained inside $f_B$ it follows that 
$$\zmax\le  z_Y(\er) =  - \omega_X \llbracket Y \rrbracket < -\omega_X \llbracket B \rrbracket \le 2 R^2.$$
\end{proof}

By rounding the $z$-coordinates we might violate the convexity of the lifting. 
The following lemma shows us how to carry out the rounding of the $z$-coordinates after we have rounded
the coordinates of the flat embedding, such that the resulting embedding remains a convex polytope.
\begin{lem}\label{lem:zperturbation}
By setting  $\alpha_z=1 /(3 R)$ and rounding such that all $z$-coordinates are multiples of $\alpha_z$ the 
lifting defined by $\zeta'_i$ and the perturbed flat embedding will remain an embedding of a convex polytope. 
\end{lem}
\begin{proof}
Let $z_i'$ be the coordinate $z_i$ after rounding and, similarly, let $\omega'_X$ be the stress after rounding the $z$-coordinates.
We have that $0 \le z_i-z_i'\le \alpha_z$. By Lemma~\ref{lem:stressperturbation} the 
interior ridges in the perturbed flat embedding have a stress $\omega_X$ that is at least $4/5$. Let $f_X$ be any interior ridge 
that is incident to the facets $f_S$ and $f_T$. By Equation~\eqref{eq:altstress} we can 
express $\omega'_X$ in terms of  $[\mathcal{T'}],\llbracket S' \rrbracket$, and $\llbracket T' \rrbracket$, where 
$\mathcal{T'}$ is formed by the point set $S'\cup T'$. The value of $[\mathcal{T}']$ can be expressed as $\sum_{i\in I} z'_i \llbracket A'_i \rrbracket$, where $I$ is the index set of the 
vertices of $\mathcal{T}'$ and $A'_i$ is given by $\mathcal{T}'\setminus\{\p'_i\}$ in some appropriate order (Laplace expansion). We split the set $I$ into $I_+:=\{i\in I \mid \llbracket A_i \rrbracket \ge 0\}$ and $I_-:=I\setminus I_+$.
As usual $f_B$ denotes the boundary face. 
Note that the projections of the simplices spanned by the sets  $A'_i$ double-cover the projection of 
$\mathcal{T}'$ into the $z=0$ hyperplane, and therefore $ \sum_{i\in I} | \llbracket A'_i \rrbracket |\le 2\llbracket B' \rrbracket \le 2\llbracket B \rrbracket \le 2R$.  Moreover, by  Corollary~\ref{cor:perturbation} and our choice of $\alpha$, 
we have $|\llbracket S'\rrbracket \llbracket T'\rrbracket |\ge (1-1/(10R))^2$. The stress on an interior ridge $f_X$  after rounding ($\omega_X'$) can be bounded as follows
\begin{align*}  \omega_X' = \frac{ [\mathcal{T}']}{\llbracket S' \rrbracket \llbracket T' \rrbracket}
& = \frac{\sum_{i\in I} z'_i \llbracket A'_i \rrbracket}{\llbracket S' \rrbracket \llbracket T' \rrbracket} 
= \frac{\sum_{i\in I_+} z'_i |\llbracket A'_i \rrbracket|}{\llbracket S' \rrbracket \llbracket T' \rrbracket} 
-\frac{\sum_{i\in I_-} z'_i |\llbracket A'_i \rrbracket|}{\llbracket S' \rrbracket \llbracket T' \rrbracket} 
\\
& \ge  \frac{\sum_{i\in I_+} (z_i-\alpha_z) |\llbracket A'_i \rrbracket|}{\llbracket S' \rrbracket \llbracket T' \rrbracket}
-\frac{\sum_{i\in I_-} z_i |\llbracket A'_i \rrbracket|}{\llbracket S' \rrbracket \llbracket T' \rrbracket} 
 \\
 & =  \underbrace{\frac{\sum_{i\in I} z_i \llbracket A'_i \rrbracket}{\llbracket S' \rrbracket \llbracket T' \rrbracket}}_{t_1}  -  
\underbrace{\frac{\sum_{i\in I_+} \alpha_z |\llbracket A'_i \rrbracket|}{\llbracket S' \rrbracket \llbracket T' \rrbracket}}_{t_2}
%
 \end{align*}
We observe that the $ \sum_{i\in I} z_i \llbracket A'_i \rrbracket= [ \mathcal{T}'_z ]$, where 
$\mathcal{T}'_z$ denotes $\mathcal{T}$ after rounding in the flat embedding, but before rounding the $z$-coordinates. This  
gives us $t_1 \ge [\mathcal{T}'_z]/(\llbracket S' \rrbracket \llbracket T' \rrbracket) = \omega_X \ge 4/5$. Moreover,
as already noticed, $ \sum_{i\in I_+}  |\llbracket A'_i \rrbracket| \le \sum_{i\in I}  |\llbracket A'_i \rrbracket| \le 2R$. After plugging in our choice for $\alpha_z$ and 
our bound for $\llbracket S'\rrbracket \llbracket T'\rrbracket$ we obtain $t_2 \le 2R /(3R (1-1/(10R))^2)$. This yields
\begin{align*}
\omega_X' \ge t_1 - t_2 \ge \frac{4}{5} - \frac{2R}{3 (1-1/(10R))^2 R} = \frac{4}{5} - \frac{200 R^2}{3(10R-1)^2}.
 \end{align*}
Note that the last expression is a monotone increasing function which is positive for $R\ge 3$. 

We are left with checking the sign for the stresses on the ridges that define the boundary of $f_B$. The stresses for these faces have to remain negative. Note that this is certainly the case, if all $z$-coordinates after the rounding remain positive. Before rounding the $z$-coordinates, every $z$-coordinate of a vertex not on $f_B$ was at least as large as the smallest vertical shift $\zeta'_i$. Since the vertical shifts are defined as the sum of two \fws, we have that the nonzero $z$-coordinates are at least $(1-1/(10R))^2$. As observed earlier, rounding the $z$-coordinates decreases the $z$-coordinates by at most $1/(3R)$. Therefore we have for every $\p_i\not \in B$
$$ z'_i\ge \left( 1- \frac{1}{10R}\right)^2- \frac{1}{3R}> 1- \frac{1}{5R} - \frac{1}{3R} >0,$$
for every $R\ge 3$. 
Hence, after rounding the $z$-coordinates the sign pattern of the stresses verifies the convexity of the perturbed realization.
\end{proof}

We now summarize our analysis and state the main theorem.
\setcounter{thm}{1}
\addtocounter{thm}{-1}
\mainthm

\begin{proof}
To get integer coordinates we multiply all coordinates after the rounding with $1/\alpha$, except 
the $z$-coordinates, which we multiply with $1/\alpha_z$. Since the maximal $z$-coordinate
is by Lemma~\ref{lem:zmax} at most $2R^2$, and we scale by $1 / \alpha_z= 3 R$ the bound for the $z$-coordinates in the 
theorem follows. All other coordinates are positive and smaller then $L$ before rounding. Hence by
scaling with $1 /\alpha =  10 d^2  L^{d-2} R$ we get that all coordinates are integers 
(Lemma~\ref{lem:stressperturbation}) and the maximum coordinate has size
$10 d^2 L^{d-1} R$. Plugging in the definition of $L$ gives as upper bound $10 d^2 R^2$
as asserted.
\end{proof}

By expressing the quantity $R$ in terms of $n$ we can restate Theorem~\ref{thm:main} as the following corollary.
\begin{cor}\label{cor:main}
For a fixed $d$, every $d$-dimensional stacked polytope can be realized on an integer grid polynomial
in $n$. The size of the largest $z$-coordinate is bounded by $O(n^{3\log(2d)})$, all other coordinates 
are bounded by $O(n^{2\log(2d)})$.
\end{cor}
Table~\ref{tab:1} lists the induced grid bounds for $d=3,\ldots, 10$.

\begin{table}[htp]
\begin{center}
\begin{tabular}{c c c }
\toprule
$d$ & exponent largest non-$z$-coordinate & exponent largest $z$-coordinate \\
\midrule
 3 & 5.17 & 7.76 \\
 4 & 6 & 9 \\
 5 & 6.65 & 9.97 \\
 6 & 7.17 & 10.76  \\
 7 & 7.62 & 11.43 \\
 8 & 8 & 12 \\
 9 & 8.34 & 12.51\\
 10 & 8.65 & 12.97 \\
 \bottomrule
\end{tabular}
\end{center}
\caption{The induced grid bounds in terms of $n$ up to dimension 10.}
\label{tab:1}
\end{table}%

\section{A simple lower bound}
\label{sec:lower}
In this section we present a simple lower bound. The basis of our construction is the following graph. Take the tetrahedron and stack a vertex in every face, then take the resulting graph and stack again a vertex in every face. We call this graph $B_3$; see Fig.~\ref{fig:b3} for an illustration.

\begin{figure}[htp]
\centering
\label{fig:b3}
	\includegraphics[scale=.6]{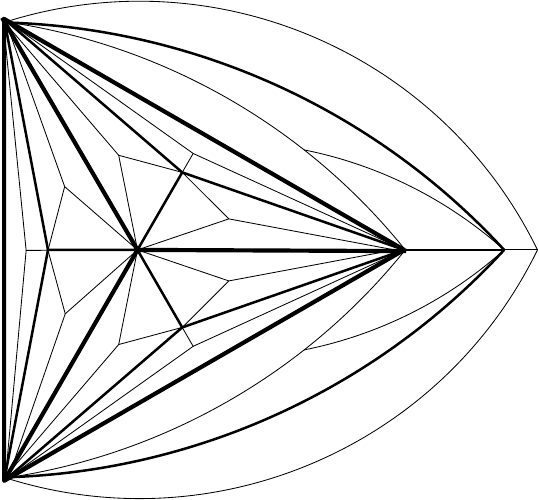}
\caption{The graph $B_3$. It has 36 faces and 20 vertices, 12 of them having degree~3.}
\end{figure}

\begin{lem}\label{lem:interiortriangle}
Let $P$ be any embedding of $B_3$ as a stacked 3-polytope. Then there is at least one face with no boundary edge in the orthogonal projection of $P$ into the $xy$-plane.	
\end{lem}
\begin{proof}
Let $\sigma$ be the boundary of $\pi(P)$, which is a polygon. We proceed with a case distinction on the size of $\sigma$.\\
{\bf Case 1:} $\sigma$ contains less than 18 vertices. \\
Every edge on $\sigma$  is a boundary edge of two faces.
Hence in this case we can have at most 34 faces with a boundary edge, but there are 36 faces in $B_3$. So we have at least two faces without a boundary edge in the projection. \\
{\bf Case 2:} $\sigma$ contains at least 18 vertices. \\
The cycle $\sigma$ splits $B_3$
into two triangulations (interior+cycle and exterior+cycle). 
Let us now have a look what happens at a degree~3 vertex of $B_3$ on $\sigma$. 
Two of its incident edges have to be in $\sigma$, 
which means it is the ``tip'' of an ear in one of the triangulations. 
Since we have 12 degree~3 vertices and at most two vertices are not on $\sigma$, one of the triangulations, say $Q$, has more than two ears. Now we look at the dual graph of $Q$ in which we removed the vertex for the outer face. Clearly, this graph is connected (since $B_3$ is 3-connected) and it has as many degree~1 vertices as $Q$ has ears. Hence, there has to be a vertex in the dual 
graph with degree~3. This means that there is a face in $Q$, whose adjacent faces are interior faces in $Q$. Therefore, we have a face in the projection without a boundary edge as asserted.
\end{proof}

There are planar 3-trees that require  $\Omega(n^2)$ area for plane straight-line drawings~\cite{MNRA11}.
Let $G_m$ be such a planar 3-tree with $m$ vertices that needs $\Omega(m^2)$ area. For simplicity we assume that $n$ is a multiple of 36. We glue a copy of $G_{n/36}$ in each of the 36 faces of $B_3$. This yields another planar 3-tree with $n$ vertices, which we call $\Gamma_n$.

\begin{lem}
The embedding of $\Gamma_n$ as a convex 3-polytope requires a bounding box of $\Omega(n^3)$ volume.
\end{lem}
\begin{proof}
Let $P$ be an embedding of $\Gamma_n$ as a 3-polytope. By restricting the 1-skeleton of $P$ to $B_3$ and taking the convex hull, we get an embedding of $B_3$ as a 3-polytope $P_B$. Due to Lemma~\ref{lem:interiortriangle}, there has to be one face, say $f$, which has no boundary edge on $\pi(P_B)$. The face $f$ defines in $\pi(P)$ a triangle, which contains the graph $G'=G_{n/36}$. The supporting planes of the faces adjacent to $f$ in $P_B$ define a cone $C$. We denote with $C_0$ the part $C\setminus P$ that contains the apex of $C$. Since none of the edges of $f$ are on the boundary of $\pi(P_B)$ we have that $\pi(C_0)=\pi(f)$. Hence, $\pi(P)$ contains a noncrossing drawing of $G_{n/36}$, namely $\pi(G')$, inside $\pi(f)$. Therefore, $\pi(P)$ needs at least area $\Omega(n^2)$.

There was nothing special with choosing the projection into the $xy$-plane. By the same arguments we can show, that also the projections into the $xz$- and $yz$-plane require $\Omega(n^2)$ area.

Let $d_x$, $d_y$, $d_z$ denote the dimensions of the bounding box of $P$ along the $x$-, $y$-, and $z$-axis. The volume of the bounding box can be estimated by 
$$d_xd_yd_z = \sqrt{d_x^2 d_y^2 d_z^2}= \sqrt{(d_xd_y)(d_xd_z)(d_yd_z)}= \Omega(n^3).$$
\end{proof}

\paragraph{Acknowledgments.} \hskip1ex \\  We thank Alexander Igamberdiev and an anonymous reviewer for helpful comments.

\bibliographystyle{alpha}
\bibliography{stacked}

\begin{thebibliography}{MNRA11}

\bibitem[And61]{A61}
George~E. Andrews.
\newblock A lower bound for the volume of strictly convex bodies with many
  boundary lattice points.
\newblock {\em Trans. Amer. Math. Soc.}, 99:272--277, 1961.

\bibitem[A{\v{Z}}95]{AZ95}
Dragan~M. Acketa and Jovisa~D. {\v{Z}}un{\'\i}{\'c}.
\newblock On the maximal number of edges of convex digital polygons included
  into an $m\times m$-grid.
\newblock {\em J. Comb. Theory Ser. A}, 69(2):358--368, 1995.

\bibitem[BR06]{BR06}
Imre B{\'a}r{\'a}ny and G{\"u}nter Rote.
\newblock Strictly convex drawings of planar graphs.
\newblock {\em Documenta Math.}, 11:369--391, 2006.

\bibitem[BS10]{BS10}
Kevin Buchin and Andr{\'e} Schulz.
\newblock On the number of spanning trees a planar graph can have.
\newblock In Mark de~Berg and Ulrich Meyer, editors, {\em Proc. Algorithms --
  ESA (1)}, volume 6346 of {\em Lecture Notes in Computer Science}, pages
  110--121. Springer, 2010.

\bibitem[CGT96]{CGT96}
Marek Chrobak, Michael~T. Goodrich, and Roberto Tamassia.
\newblock Convex drawings of graphs in two and three dimensions (preliminary
  version).
\newblock In {\em Proc. 12th Symposium on Computational Geometry (SoCG)}, pages
  319--328, 1996.

\bibitem[DG97]{DG97}
Gautam Das and Michael~T. Goodrich.
\newblock On the complexity of optimization problems for 3-dimensional convex
  polyhedra and decision trees.
\newblock {\em Computational Geometry: Theory and Applications}, 8(3):123--137,
  1997.

\bibitem[DLRS10]{drs10}
Jesus De~Loera, J{\"o}rg Rambau, and Francisco Santos.
\newblock {\em Triangulations}, volume~25 of {\em Algorithms and Computation in
  Mathematics}.
\newblock Springer, 2010.

\bibitem[DS11]{DS11}
Erik~D. Demaine and Andr{\'e} Schulz.
\newblock Embedding stacked polytopes on a polynomial-size grid.
\newblock In {\em Proc. 22nd ACM-SIAM Symposium on Discrete Algorithms (SODA),
  San Francsico, 2011}, pages 1177--1187. ACM Press, 2011.

\bibitem[EG95]{EG95}
Peter Eades and Patrick Garvan.
\newblock Drawing stressed planar graphs in three dimensions.
\newblock In Franz-Josef Brandenburg, editor, {\em Graph Drawing}, volume 1027
  of {\em Lecture Notes in Computer Science}, pages 212--223. Springer, 1995.

\bibitem[HK92]{HK92}
John~E. Hopcroft and Peter~J. Kahn.
\newblock A paradigm for robust geometric algorithms.
\newblock {\em Algorithmica}, 7(4):339--380, 1992.

\bibitem[IS16]{IS13}
Alexander Igamberdiev and Andr{\'{e}} Schulz.
\newblock A duality transform for constructing small grid embeddings of 3d
  polytopes.
\newblock {\em Comput. Geom.}, 56:19--36, 2016.

\bibitem[Lov00]{L00}
L\'aszl\'o Lov\'asz.
\newblock Steinitz representations of polyhedra and the {C}olin de
  {V}erdi{\`e}re number.
\newblock {\em J. Comb. Theory, Ser. B}, 82:223--236, 2000.

\bibitem[MNRA11]{MNRA11}
Debajyoti Mondal, Rahnuma~Islam Nishat, Md.~Saidur Rahman, and
  Muhammad~Jawaherul Alam.
\newblock Minimum-area drawings of plane 3-trees.
\newblock {\em Journal of Graph Algorithms and Applications}, 15(2):177--204,
  2011.

\bibitem[MRS11]{RRS11}
Ares~Rib{\'o} Mor, G{\"u}nter Rote, and Andr{\'e} Schulz.
\newblock Small grid embeddings of 3-polytopes.
\newblock {\em Discrete {\&} Computational Geometry}, 45(1):65--87, 2011.

\bibitem[OS94]{OS94}
Shmuel Onn and Bernd Sturmfels.
\newblock A quantitative {S}teinitz' theorem.
\newblock In {\em Beitr{\"a}ge zur Algebra und Geometrie}, volume~35, pages
  125--129, 1994.

\bibitem[PW13]{PW13}
Igor Pak and Stedman Wilson.
\newblock A quantitative {S}teinitz theorem for plane triangulations.
\newblock \url{http://arxiv.org/abs/1311.0558}, 2013.

\bibitem[RG96]{RG96}
J\"urgen Richter-Gebert.
\newblock {\em Realization Spaces of Polytopes}, volume 1643 of {\em Lecture
  Notes in Mathematics}.
\newblock Springer, 1996.

\bibitem[RGZ95]{GZ95}
J\"urgen Richter-Gebert and G\"{u}nter~M. Ziegler.
\newblock Realization spaces of 4-polytopes are universal.
\newblock {\em Bull. Amer. Math. Soc.}, 32:403, 1995.

\bibitem[{Rib}06]{R06}
Ares {Rib\'o Mor}.
\newblock {\em Realization and Counting Problems for Planar Structures: Trees
  and Linkages, Polytopes and Polyominoes}.
\newblock PhD thesis, Freie Universit\"at Berlin, 2006.

\bibitem[Ryb99]{R99}
Konstantin~A. Rybnikov.
\newblock Stresses and liftings of cell-complexes.
\newblock {\em Discrete {\&} Computational Geometry}, 21(4):481--517, 1999.

\bibitem[Sch91]{S73}
Oded Schramm.
\newblock Existence and uniqueness of packings with specified combinatorics.
\newblock {\em Israel J. Math.}, 73:321--341, 1991.

\bibitem[Sch11]{Sch11}
Andr{\'e} Schulz.
\newblock Drawing 3-polytopes with good vertex resolution.
\newblock {\em Journal of Graph Algorithms and Applications}, 15(1):33--52,
  2011.

\bibitem[Ste22]{S22}
Ernst Steinitz.
\newblock Polyeder und {Raumeinteilungen}.
\newblock In {\em Encyclop\"adie der mathematischen {W}issenschaften}, volume
  3-1-2 ({Geometrie}), chapter~12, pages 1--139. B.~G. Teubner, Leipzig, 1922.

\bibitem[Tar83]{T83}
Robert~Endre Tarjan.
\newblock Linking and cutting trees.
\newblock In {\em Data Structures and Network Algorithms}, chapter~5, pages
  59--70. Society for Industrial and Applied Mathematics, 1983.

\bibitem[Thi91]{T91}
Torsten Thiele.
\newblock Extremalprobleme f\"ur {P}unktmengen.
\newblock Master's thesis, Freie Universit\"at Berlin, 1991.

\bibitem[Tut60]{T60}
William~T. Tutte.
\newblock Convex representations of graphs.
\newblock {\em Proceedings London Mathematical Society}, 10(38):304--320, 1960.

\bibitem[Tut63]{T63}
William~T. Tutte.
\newblock How to draw a graph.
\newblock {\em Proceedings London Mathematical Society}, 13(52):743--768, 1963.

\bibitem[Zic07]{Z09}
Florian Zickfeld.
\newblock {\em Geometric and Combinatorial Structures on Graphs}.
\newblock PhD thesis, Technische Universit{\"a}t Berlin, December 2007.

\bibitem[Zie95]{Z95}
G{\"u}nter~M. Ziegler.
\newblock {\em Lectures on Polytopes}.
\newblock Springer, 1995.

\end{thebibliography}

\end{document}